\documentclass[pdflatex,sn-mathphys-num]{sn-jnl}%
\usepackage{environ}%
\makeatletter
\let\SNabstract\abstract\let\abstract\relax
\NewEnviron{abstract}{\expandafter\SNabstract\expandafter{\expandafter\unboldmath\BODY}\global\let\@abstract\@abstract}

\RequirePackage{amsmath,amssymb,amsthm}
\RequirePackage{etoolbox}
\RequirePackage{aliascnt}%
\RequirePackage{tikz}
\RequirePackage{enumitem}%
\setlist[enumerate]{label=\textup{(\roman*)},ref=\textup{(\roman*)},leftmargin=*,topsep=2pt,itemsep=1pt}
\AtEndPreamble{%
  \RequirePackage[capitalize,noabbrev]{cleveref}%
  \crefname{theorem}{Theorem}{Theorems}%
  \crefname{lemma}{Lemma}{Lemmas}%
  \crefname{proposition}{Proposition}{Propositions}%
  \crefname{corollary}{Corollary}{Corollaries}%
  \crefname{definition}{Definition}{Definitions}%
  \crefname{fact}{Fact}{Facts}%
  \crefname{remark}{Remark}{Remarks}%
  \crefname{example}{Example}{Examples}%
  \crefname{figure}{Figure}{Figures}%
  \theoremstyle{thmstyleone}%
  \newtheorem{theorem}{Theorem}%
  \crefname{introthm}{Theorem}{Theorems}%
  \newtheorem{introthm}{Theorem}%
  \newaliascnt{lemma}{theorem}\newtheorem{lemma}[lemma]{Lemma}\aliascntresetthe{lemma}%
  \newaliascnt{proposition}{theorem}\newtheorem{proposition}[proposition]{Proposition}\aliascntresetthe{proposition}%
  \newaliascnt{corollary}{theorem}\aliascntresetthe{corollary}%
  \newaliascnt{fact}{theorem}\newtheorem{fact}[fact]{Fact}\aliascntresetthe{fact}%
  \theoremstyle{thmstylethree}%
  \newaliascnt{definition}{theorem}\newtheorem{definition}[definition]{Definition}\aliascntresetthe{definition}%
  \theoremstyle{thmstyletwo}%
  \newaliascnt{example}{theorem}\newtheorem{example}[example]{Example}\aliascntresetthe{example}%
  \newtheorem*{example*}{Example}%
  \newaliascnt{remark}{theorem}\newtheorem{remark}[remark]{Remark}\aliascntresetthe{remark}%
}

\newcommand{\sched}{H}            %
\newcommand{\comp}{G}             %
\newcommand{\evol}{U}             %
\newcommand{\devprod}{W}          %
\newcommand{\compevol}{V}         %

\newcommand{\bits}{k}             %
\newcommand{\sparsity}{\sigma}    %
\newcommand{\width}{w}            %
\newcommand{\freearc}{\vartheta}  %

\newcommand{\gap}[1]{\Gamma(#1)}   %

\newcommand{\cormult}{\mathcal N} %

\hypersetup{hidelinks}%
\ifdefined\pdfvariable
  \pdfvariable suppressoptionalinfo 1023 \pdfvariable trailerid {}
\else\ifdefined\pdfsuppressptexinfo
  \pdfsuppressptexinfo=-1 \pdfinfoomitdate=1 \pdftrailerid{}
\fi\fi
\hypersetup{pdfcreator={},pdfproducer={}}
\hypersetup{pdftitle={Rounded Clocks and Finite Precision: A Transfer Principle Between Continuous and Discrete Quantum Dynamics},
  pdfauthor={Michael Jarret}}
\newcommand{\leanrepo}{https://github.com/mjarret/products-of-unitaries-lean}
\newcommand{\leanrepolink}{\expandafter\url\expandafter{\leanrepo}}
\DeclareRobustCommand{\lean}[2]{\href{\leanrepo/blob/v1.0.0/#1\#L#2}{Lean}}

\AtEndPreamble{\counterwithin{theorem}{section}\counterwithin{equation}{section}%
  \@addtoreset{lemma}{section}\@addtoreset{proposition}{section}\@addtoreset{corollary}{section}%
  \@addtoreset{fact}{section}\@addtoreset{definition}{section}\@addtoreset{remark}{section}\@addtoreset{example}{section}%
  \aliascntresetthe{lemma}\aliascntresetthe{proposition}\aliascntresetthe{corollary}%
  \aliascntresetthe{fact}\aliascntresetthe{definition}\aliascntresetthe{remark}\aliascntresetthe{example}%
  }
\makeatother
\makeatletter
\newcommand\dedication[1]{\gdef\@dedication{#1}}
\let\sn@printabstract\printabstract
\def\printabstract{\ifdefined\@dedication{\centering\itshape\@dedication\par}\vskip12pt\fi\sn@printabstract}
\makeatother

\begin{document}

\title{Rounded Clocks and Finite Precision: A Transfer Principle Between Continuous and Discrete Quantum Dynamics}%
\author*[1]{\fnm{Michael} \sur{Jarret}}\email{mjarretb@gmu.edu}
\affil*[1]{\orgdiv{Department of Mathematical Sciences and Department of Computer Science},
\orgname{George Mason University}, \orgaddress{\city{Fairfax}, \state{VA}, \postcode{22030}, \country{USA}}}%

\dedication{To Farzaneh Sabbagh, who had known me for three months when I was badly injured in a car accident, saw me through a long recovery with more patience than I had any right to ask, and then married me. This paper would not exist without her.}%

\begin{abstract}
I present a short approach to discretizing continuous unitary evolution in finite dimensions. For a time-dependent Hamiltonian $H(s)$ with $s\in[0,1]$ and an integer $T\ge1$, let $U(T)$ be the evolution of $H(t/T)$ over time $T$, and let $W_T=e^{-iH((T-1)/T)}\cdots e^{-iH(0)}$ be the product of the unitaries generated by its samples. With a half-step phase at each end, $V(T)=e^{-iH(1)/2}W_Te^{iH(0)/2}$ is the evolution generated by $H(\lfloor t+1/2\rfloor/T)$, where $\lfloor t+1/2\rfloor$ is a rounded clock, time at finite precision. Whereas Duhamel's formula only bounds $\|V(T)-U(T)\|$ by a constant, I show that $\|V(T)-U(T)\|=O(1/T)$ in operator norm, with explicit constants, whenever $H'$ is Lipschitz and every sample has eigenvalue differences at most a fixed width below $2\pi$. The gain comes from cancellation within and between time steps. This general transfer theorem assumes neither adiabaticity nor a spectral gap, and any operator-norm bound for $U(T)$ transfers to $V(T)$ with an added $O(1/T)$. Nonetheless, discrete adiabatic theorems for $W_T$ follow from continuous ones, because the end phases commute with the spectral projectors of $H(0)$ and $H(1)$. For a product of samples of a $C^2$ unitary family, a continuously moving arc of fixed length $\vartheta$ that avoids the spectrum determines a Hermitian logarithm meeting the width condition, and the $O(1/T)$ bound has coefficients $O(\vartheta^{-3})$. In the easier direction, Duhamel's formula shows that rounding the clock to midpoints and truncating the entries of a regulated time-dependent Hamiltonian at a fixed binary place gives a product within $T$ times the error of the rounded, truncated Hamiltonian.
\end{abstract}
 
\keywords{discretization of unitary evolution, Hamiltonian simulation, product formulas, discrete time, adiabatic theorem, Duhamel formula, binary expansion}
\pacs[MSC Classification]{81Q15, 65L05, 47A55}

\maketitle

\section{Discretizing unitary evolution}\label{sec:intro}\label{sec:intro-results}

In quantum dynamics, we describe evolution either continuously, through a time-dependent Hamiltonian, a family of Hermitian matrices $H(t)$ whose evolution $U(t)$ solves Schr\"odinger's equation $i\,dU/dt=H(t)U(t)$ with $U(0)=1$, or discretely, through a product of unitaries. A finite product of evolutions under time-independent Hamiltonians is itself generated by a piecewise-constant Hamiltonian. A continuous evolution generated by $H(t)$ is therefore approximated by a finite product whenever a piecewise-constant Hamiltonian $G(t)$ generates approximately the same evolution. Storing $H$ at finitely many times, with each entry truncated to finitely many bits, defines such a $G$, the truncated $H$ on a rounded clock that keeps time at finite precision.

Duhamel's formula compares two such evolutions, but only at short steps. For a time-dependent Hamiltonian $\sched(s)$ with $s\in[0,1]$, the evolution $\evol$ over time $T$ is the one generated by $\sched(t/T)$ for $t\in[0,T]$, and Duhamel's formula bounds its distance from the evolution of $G$ by $\int_0^T\|\sched(t/T)-G(t)\|\,dt$. Many products, however, keep a fixed step length as $T$ grows, among them digital simulations that apply, for each unit of time, the evolution under one fixed Hamiltonian, the walks of discrete adiabatic theorems \cite{costa2022,an2025}, and the circuits of digital adiabatic state preparation \cite{yi2021,kovalsky2023}. The simplest, for an integer $T\ge1$, is the product $\devprod_T=e^{-i\sched((T-1)/T)}\cdots e^{-i\sched(0)}$ of the exponentials of the samples, whose steps have length one. Its generator differs from $\sched(t/T)$ by $O(1/T)$ for a time $T$, so Duhamel's formula bounds its distance from $\evol(T)$ only by a constant. Fixed-step bounds that do not grow with $T$ are known only for one initial state, an eigenstate isolated by a spectral gap, each through its own analysis of the product \cite{yi2021,kovalsky2023,an2025,lu2026}. I ask whether such separate analyses, and the Dyson and Magnus expansions behind product formulas and integrators, are necessary, or whether a simple transfer principle carries continuous bounds to products. That principle would be an operator-norm comparison at a fixed step length, for every initial state and with no spectral gap, that decays as $T$ grows when the values of $\sched$ do not commute.

I prove such a comparison. In its statement, we write $\|\sched'\|_\infty$ for $\max_{s\in[0,1]}\|\sched'(s)\|$.
\begin{introthm}[informal; see \Cref{thm:jumps} and \Cref{fact:pi}; \lean{UnitaryProducts/PaperProof/IntroTheorems.lean}{34}]\label{introthm:unit}
Let $T\ge1$ be an integer, let $\sched(s)$ for $s\in[0,1]$ be a time-dependent Hamiltonian whose derivative $\sched'$ has Lipschitz constant $L$, and let $\evol$ be its evolution over time $T$. If every sample $\sched(j/T)$ with $0\le j\le T$ has eigenvalue differences of modulus at most $\pi$, then
\[
 \bigl\|e^{-i\sched(1)/2}\devprod_Te^{i\sched(0)/2}-\evol(T)\bigr\|\le\frac{0.116\|\sched'\|_\infty+0.143L+0.061\|\sched'\|_\infty^2}T.
\]
\end{introthm}
The comparison rests on the rounded clock. With steps of length one, the rounded clock $\lfloor t+1/2\rfloor$ is the integer nearest $t$, and the evolution $\compevol$ generated by $\sched(\lfloor t+1/2\rfloor/T)$ satisfies $\compevol(T)=e^{-i\sched(1)/2}\devprod_Te^{i\sched(0)/2}$, the product with a half-step phase at each end (\Cref{sec:prelim}). \Cref{introthm:unit} therefore compares the evolutions of $\sched$ on the true clock and on the rounded one, and a product of unitaries is the continuous evolution at finite resolution in time.

This operator-norm comparison is the missing link between continuous and discrete evolution. Through it, any operator-norm bound proved for $\evol(T)$ holds for $\compevol(T)$, under its own hypotheses and those of the comparison, up to an explicit additive $O(1/T)$. Although the comparison assumes no adiabaticity, continuous adiabatic theorems \cite{kato1950,avron1987,jansen2007} become discrete ones for $\devprod_T$ itself, because the half-step phases commute with the spectral projectors of $\sched(0)$ and $\sched(1)$ (\Cref{sec:unit-transfer}). Among them are Guo and An's \cite{guo2025} for twice continuously differentiable gap-adapted schedules and Avron and Elgart's \cite{avronelgart1999} through eigenvalue crossings (\Cref{sec:examples}). For products of given unitaries, the comparison holds through a Hermitian logarithm whenever the unitaries come from a $C^2$ family whose spectra avoid an arc of fixed length $\freearc>0$ that moves continuously, with coefficients $O(\freearc^{-3})$ (\Cref{sec:walks}), the order in which an eigenphase gap $\freearc$ enters the Jansen--Ruskai--Seiler bound. A continuous theorem, not a separate discrete analysis, thus carries each of these conclusions to the product.

The comparison needs one spectral condition, a nonresonance condition on the samples. Its formal version \eqref{eq:products} allows eigenvalue differences of modulus up to any fixed width $\width<2\pi$. The condition is natural, because at exact resonance, where eigenvalues $\lambda_p$ and $\lambda_q$ of a sample $\sched(j/T)$ differ by a nonzero multiple of $2\pi$, the correction that the proof needs exists only if $\sched'(j/T)_{pq}=0$ in an orthonormal eigenbasis of $\sched(j/T)$ (\Cref{sec:comparison}). It is general, because evolution depends only on the product of energy and time, so in simulation the energy scale is a gauge choice, and lowering it while lengthening time in proportion lets every bounded time-dependent Hamiltonian meet the condition at steps of length one. The proof is short and self-contained, one integration by parts and, at each sample, one linear equation solved entry by entry.

One comparison theorem, Duhamel's formula integrated by parts against a correction $Y$ that may jump (\Cref{thm:comparison}), gives both directions of the correspondence between products and evolutions, but the two directions are not symmetric. Approximating $\evol(T)$ by a product is the easy one (\Cref{sec:results}). There we take $Y=0$, so that the comparison theorem is Duhamel's bound. We write $\sparsity_\bits$ for the largest number of entries that truncation at the $\bits$th binary place changes in one row of $\sched(s)$, over all rows and all $s\in[0,1]$.
\begin{introthm}[informal; see \Cref{thm:fine-steps} and \eqref{eq:fine-bits}; \lean{UnitaryProducts/PaperProof/OscillationFineSteps.lean}{296}]\label{introthm:fine-count}
Let $\sched(s)$ for $s\in[0,1]$ be a regulated time-dependent Hamiltonian with evolution $\evol$ over time $T>0$, and let $\varepsilon>0$. Take checkpoints such that, inside each interval between them, $\sched(s)$ is within $\varepsilon$ of the sample at the interval's midpoint, which every regulated $\sched$ admits. Inside each interval, replace $\sched$ by that sample with every entry's real and imaginary parts truncated toward zero at the $\bits$th binary place. The resulting product is within $T(\varepsilon+\sqrt2\,\sparsity_\bits\,2^{-\bits})$ of $\evol(T)$.
\end{introthm}
The easy direction is general because it uses Duhamel's formula alone. The bound needs only that truncation keeps each sample Hermitian, so any odd function of the entries may replace bit truncation, and that $\sched$ is regulated, with one-sided limits everywhere, so $\sched$ may jump. Hamiltonian simulation designs a product and then bounds its error \cite{wiebe2011,poulin2011}, whereas here the product is the one that storing $\sched$ at finite precision already defines. \Cref{introthm:fine-count} reaches $O(1/T)$ only with $\varepsilon$ of order $1/T^2$, or about $\|\sched'\|_\infty T^2$ factors, against $T$ for \Cref{introthm:unit}. At short steps, Duhamel's formula thus transfers continuous bounds to products with no Dyson or Magnus expansion.

Identifying a product at steps of fixed length with a continuous evolution is the hard direction (\Cref{sec:coarse}). Even for $\compevol(T)$, Duhamel's formula only gives $\|\compevol(T)-\evol(T)\|\le\|\sched'\|_\infty/4$, as the clock error $|t-\lfloor t+1/2\rfloor|$ has mean $1/4$. The signed error $t-\lfloor t+1/2\rfloor$ changes sign within each step, however, and we keep that cancellation by integrating by parts before taking any norm, against corrections $Y_j$ that solve $Y_j'+i[\sched(j/T),Y_j]=(t-j)\sched'(j/T)/T$ for $|t-j|<\frac12$ with the same value at $j-\frac12$ and $j+\frac12$ (\Cref{thm:comparison,thm:clocks}). That value changes by $O(1/T^2)$ from one sample to the next, so the $T$ jumps between neighboring corrections contribute only $O(1/T)$, which gives \Cref{introthm:unit} (\Cref{thm:jumps}).

\subsection{Prior work}
Wiebe et al.\ \cite{wiebe2011} and Poulin et al.\ \cite{poulin2011} simulate continuous evolutions on a quantum computer by products of unitaries, and Watkins et al.\ \cite{watkins2024} do so through a discrete clock register, a quantum degree of freedom that removes the time ordering. Blanes, Casas, and Murua \cite{blanes2008}, Hochbruck and Lubich \cite{hochbruck2003}, and Blanes et al.\ \cite{blanes2009} analyze the splitting, exponential, and Magnus integrators that do the same in numerical analysis. Burgarth et al.\ \cite{burgarth2022} represent product formulas as evolutions of piecewise-constant Hamiltonians, and Yi \cite{yi2021} reads each Trotter step as the exact evolution of an effective Hamiltonian.

Two of these analyses apply directly under the hypotheses of \Cref{introthm:unit}. Hochbruck and Lubich \cite[Theorem~3.1]{hochbruck2003} analyze exponential midpoint steps with no bound on step length times Hamiltonian norm and no spectral gap, and away from its two end half steps $\compevol$ is the exponential midpoint rule with step one. They bound the error by summing norms of local errors over the steps, which loses any cancellation between neighboring steps that my corrections keep, and specialized to $\sched(t/T)$ at steps of length one, the commutator terms of order $1/T$ in their local error bounds sum to $O(1)$ when the values of $\sched$ do not commute. In its choice of $Y$, \Cref{thm:comparison} generalizes the lemma of Burgarth et al., which takes $Y$ to be the integral of the Hamiltonian difference and allows locally integrable bounded generators on a separable Hilbert space. For $\sched(t/T)$ and $\sched(\lfloor t+1/2\rfloor/T)$, that integral is $O(1/T)$ at every time, but their bound multiplies it by one plus the integrals over $[0,T]$ of the norms of the two Hamiltonians, which grow linearly in $T$, so their bound is $O(1)$ as well.

Conjugating the steps of a method by one map is processing, which Blanes, Casas, and Murua survey \cite[Section~4.2]{blanes2008}. Layden \cite{layden2022} bounds first-order Trotter errors for a time-independent Hamiltonian through conjugation by a half step, and the half-step phases process $\devprod_T$ in the same way. Wisdom, Holman, and Touma \cite{wisdom1996} read a splitting map as the evolution of a Hamiltonian with periodic delta kicks, much as the rounded clock lets one read a product as an evolution, and correct the map by a formal series whose divisors vanish at step-size resonances. My corrections play the role of their symplectic corrector, as exact solutions under the width hypothesis, and are relatives of the kick operator of Floquet theory. Floquet expansions in the inverse frequency \cite{rahav2003,bukov2015}, rigorous effective Hamiltonians at high frequency \cite{kuwahara2016,abanin2017}, and the modified equations of backward error analysis \cite{hairer2006} concern periodic Hamiltonians or short steps, whereas the comparison here is global and in operator norm, for a slowly varying $\sched$ at a fixed step. The integrators of Jahnke and Lubich \cite{jahnke2003} take long steps in an eigenbasis of the Hamiltonian, with a complete error analysis in the adiabatic case and smaller steps near avoided crossings, whereas my only spectral condition is the width hypothesis.

Discrete adiabatic theorems \cite{dranov1998,tanaka2011,costa2022,an2025} take a product as given and bound its distance from adiabatic transport by separate arguments, under an eigenphase (eigenvalue-argument) gap for a tracked band. Tanaka \cite{tanaka2011} notes that the proof of the adiabatic theorem for slowly modulated Hamiltonians, through eigenvectors of a Floquet Hamiltonian, does not apply to quantum maps in general. For walks, Costa et al.\ \cite{costa2022} and An, Costa, and Berry \cite{an2025} prove $O(1/T)$ leakage under eigenphase gaps across neighboring samples, and An, Costa, and Berry also deduce $O(1/T)$ error at a fixed step from a spectral gap of the Hamiltonian, for preparing one eigenstate. For one initial state and sufficiently short fixed steps, Yi \cite{yi2021} obtains $O(1/T)$ leakage for an eigenstate of a nondegenerate spectrum, Kovalsky et al.\ \cite{kovalsky2023} for a gapped ground state, and Lu et al.\ \cite{lu2026} asymptotically for the ground state.

Unlike these bounds, \eqref{eq:products} holds in operator norm, and so for every initial state, and it compares the product with the Hamiltonian's own evolution rather than with ideal adiabatic transport, so it needs no band gap, its constants are independent of any gap, and it allows eigenvalue crossings. A gap enters only through the continuous theorem being transferred, with the dependence that theorem has. For products of given unitaries, the constants depend instead on the arc length $\freearc$, which \Cref{cor:walk-leakage} takes to be the least eigenphase gap of a gapped band. The width hypothesis by itself neither implies nor follows from eigenphase gaps of a band of the sampled exponentials, but when $\sched(s)$ has eigenvalue differences of modulus at most $\width$, the image of a band of $\sched(s)$ is a band of $e^{-i\sched(s)}$ with eigenphase gap at least the smaller of $2\pi-\width$ and the band's spectral gap.

\subsection{Outline}
In \Cref{sec:prelim}, we prove the comparison theorem, introduce the rounded clock, prove leakage transfer, and define bands. \Cref{sec:results} proves \Cref{introthm:fine-count}, and \Cref{sec:coarse} proves \Cref{introthm:unit} and its version for products of given unitaries. \Cref{sec:examples} works the examples above and a precessing spin, which shows that the endpoint coefficient $c_0$ of \Cref{thm:clocks} is optimal at every width in $(0,2\pi)$, and \Cref{sec:outlook} discusses higher-order transfer, resonance, products of arbitrary sequences of unitaries, such as the circuits of the quantum approximate optimization algorithm, and other open questions. \Cref{app:spectral} proves the three lemmas behind \Cref{sec:coarse}, the most computational parts of the paper, because I keep their constants explicit, and \Cref{app:constants} collects the constants of \Cref{sec:coarse}.
\section{Preliminaries}\label{sec:prelim}

All comparisons of evolutions in this paper come from one identity, \Cref{thm:comparison}.
Throughout, $d\ge1$, matrices are $d\times d$, $\|\cdot\|$ is the operator norm, and $\|f\|_\infty=\sup_{s\in[0,1]}\|f(s)\|$ for a matrix-valued function $f(s)$ with $s\in[0,1]$. A matrix-valued function $\comp(t)$ with $t\in[a,b]$ is \emph{piecewise constant} if there are times $a=t_0<t_1<\cdots<t_n=b$ and matrices $\comp_0,\ldots,\comp_{n-1}$ with $\comp(t)=\comp_j$ for $t\in(t_j,t_{j+1})$.
The comparison of evolutions requires the following regularity.

\begin{definition}[\lean{UnitaryProducts/Paper/Preliminaries.lean}{155}]\label{def:regulated}
A matrix-valued function $\sched(t)$ with $t\in[a,b]$ is \emph{regulated} if $\sched$ has a right limit at every $t\in[a,b)$ and a left limit at every $t\in(a,b]$.
\end{definition}

Lipschitz functions are regulated, but regulated functions may additionally jump. A regulated function $\sched(t)$ with $t\in[a,b]$ is a uniform limit of piecewise-constant functions \cite[(7.6.1)]{dieudonne1960}, so it is bounded and continuous outside a countable set.

A Hermitian family $\comp(t)$ with $t\in[a,b]$ \emph{generates} a continuous, unitary-valued \emph{evolution} $\compevol(t)$ for $t\in[a,b]$ if
\begin{equation}\label{eq:evolution}
 \begin{cases}
  i\compevol'(t)=\comp(t)\compevol(t) & \text{for all but countably many } t,\\
  \compevol(a)=1.
 \end{cases}
\end{equation}
For a regulated Hermitian family $\comp(t)$ with $t\in[a,b]$, the real-linear maps $X\mapsto-i\comp(t)X$ are regulated in $t$, so there is exactly one continuous matrix-valued $\compevol$ with $\compevol(a)=1$ and $i\compevol'=\comp\compevol$ for all but countably many $t$ \cite[(10.6.3), (10.4.6), (10.5.4)]{dieudonne1960}. The solution $\compevol$ is unitary as $\comp$ is Hermitian, so $\comp$ generates exactly one evolution. A \emph{time-dependent Hamiltonian} is a Hermitian family $\sched(s)$ with $s\in[0,1]$. Its \emph{evolution} over time $T>0$ is the evolution $\evol(t)$ for $t\in[0,T]$ generated by $\sched(t/T)$. With $s=t/T$, $\evol(Ts)$ for $s\in[0,1]$ is generated by $T\sched(s)$, so $\evol(T)$ depends only on $T\sched$. Primes denote derivatives of a function in its own variable, as in $\compevol'(t)$ for $t\in[a,b]$, $\sched'(s)$ for $s\in[0,1]$, and $f'(x)$ for a scalar function $f$.

\subsection{Comparing two evolutions}\label{sec:comparing}
Let regulated Hermitian families $K$ and $G$ generate $U$ and $V$, respectively. Wherever $U$ and $V$ satisfy \eqref{eq:evolution}, $(U^\dagger)'=iU^\dagger K$, as $K$ is Hermitian, and $V'=-iGV$. Duhamel's formula integrates the resulting derivative of $U^\dagger V$. To compare the evolutions more finely, we integrate by parts, as in proofs of the adiabatic theorem \cite{avron1987,jansen2007}, against a correction $Y$. Wherever $Y$ is also differentiable, adding zero gives
\begin{align}
 \frac d{dt}\bigl(U^\dagger V\bigr)&=iU^\dagger(K-G)V\notag\\
 &=iU^\dagger(K-G)V-i\frac d{dt}\bigl(U^\dagger YV\bigr)+i\frac d{dt}\bigl(U^\dagger YV\bigr)\notag\\
 &=iU^\dagger(K-G)V-i\bigl(iU^\dagger KYV+U^\dagger Y'V-iU^\dagger YGV\bigr)+i\frac d{dt}\bigl(U^\dagger YV\bigr)\notag\\
 &=iU^\dagger\underbrace{\bigl(K-G-Y'-i(KY-YG)\bigr)}_{R}V+i\frac d{dt}\bigl(U^\dagger YV\bigr),\label{eq:residual}
\end{align}
which defines the \emph{residual} $R$. Integrating \eqref{eq:residual} gives the error $U^\dagger V-1$ as the integral of the residual plus boundary terms (\Cref{thm:comparison}). The unit-step bound in \Cref{sec:coarse} uses one correction for each sample (\Cref{sec:comparison}), and neighboring ones can disagree where they meet, so the comparison theorem lets the correction jump. We write $[f]_a^b=f(b)-f(a)$.\footnote{For continuous $K$ and $G$, taking $m=0$ and $Y_0(t)=\int_a^t(K-G)\,ds$ in \Cref{thm:comparison} gives the integration-by-parts lemma of Burgarth et al.\ \cite{burgarth2022}.}

\begin{theorem}[Comparison of evolutions; \lean{UnitaryProducts/PaperProof/ThmComparison.lean}{43}]\label{thm:comparison}
Let regulated Hermitian families $K(t)$ and $G(t)$ with $t\in[a,b]$ generate $U$ and $V$, respectively, let $m\ge0$ be an integer, and let $a=p_0\le p_1\le\cdots\le p_{m+1}=b$. For $0\le j\le m$, let the pieces $Y_j(t)$ for $t\in[p_j,p_{j+1}]$ be continuous and matrix-valued, with a derivative that, for all but countably many $t$, exists and agrees with a regulated function. Let $R(t)$ be the residual \eqref{eq:residual} of $Y_j$ wherever $p_j<t<p_{j+1}$ and $Y_j'(t)$ exists, with arbitrary values at the remaining countably many $t$, and let $\Delta_j=Y_j(p_j)-Y_{j-1}(p_j)$ for $1\le j\le m$.
Then
\[
 U(b)^\dagger V(b)-1=\int_a^biU^\dagger RV\,dt+iU(b)^\dagger Y_m(b)V(b)-iY_0(a)-i\sum_{j=1}^mU(p_j)^\dagger\Delta_jV(p_j).
\]
\end{theorem}

\begin{proof}
For $p_j<t<p_{j+1}$, outside a countable set, \eqref{eq:residual} holds with $Y_j$ in place of $Y$. The functions $U^\dagger V$ and $U^\dagger Y_jV$ are continuous, and their derivatives agree with regulated functions outside a countable set, so each is the integral of its derivative \cite[Section~8.7]{dieudonne1960}. Likewise, $U^\dagger RV$ agrees with a regulated function outside a countable set, and its integral is that function's integral, whatever the values of $R$ at the remaining countably many $t$. Integrating \eqref{eq:residual} over each $[p_j,p_{j+1}]$ and summing over $j$ gives
\begin{align*}
 U(b)^\dagger V(b)-1&=\sum_{j=0}^m\int_{p_j}^{p_{j+1}}\frac d{dt}\bigl(U^\dagger V\bigr)\,dt\\
 &=\sum_{j=0}^m\int_{p_j}^{p_{j+1}}\Bigl(iU^\dagger RV+i\frac d{dt}\bigl(U^\dagger Y_jV\bigr)\Bigr)\,dt\\
 &=\int_a^biU^\dagger RV\,dt+i\sum_{j=0}^m\bigl[U^\dagger Y_jV\bigr]_{p_j}^{p_{j+1}}\\
 &=\int_a^biU^\dagger RV\,dt+iU(b)^\dagger Y_m(b)V(b)-iY_0(a)-i\sum_{j=1}^mU(p_j)^\dagger\Delta_jV(p_j).
\end{align*}
The first equality holds as $U(a)=V(a)=1$, the second is \eqref{eq:residual}, the third integrates the derivatives of $U^\dagger Y_jV$, and the last regroups the sum at each $p_j$, again using $U(a)=V(a)=1$.
\end{proof}

With no correction, $m=0$ and $Y_0=0$, the residual is $R=K-G$, and as $U$ and $V$ are unitary, \Cref{thm:comparison} gives \emph{Duhamel's bound}
\begin{align}
 \|V(b)-U(b)\|&=\|U(b)^\dagger V(b)-1\|\notag\\
 &=\Bigl\|\int_a^biU^\dagger(K-G)V\,dt\Bigr\|\notag\\
 &\le\int_a^b\|K-G\|\,dt,\label{eq:duhamel}
\end{align}
used in \Cref{sec:results,sec:coarse}.

A product of unitaries enters the comparison theorem as the evolution generated by a piecewise-constant Hermitian family. If $\comp$ is piecewise constant, the recursion $\compevol(a)=1$, $\compevol(t)=e^{-i(t-t_j)\comp_j}\compevol(t_j)$ for $t\in[t_j,t_{j+1}]$ defines a continuous, unitary $\compevol$ that satisfies $i\compevol'(t)=\comp(t)\compevol(t)$ for every $t\notin\{t_0,\ldots,t_n\}$, so
\begin{equation}\label{eq:piecewise-product}
  \compevol(b)=e^{-i(t_n-t_{n-1})\comp_{n-1}}\cdots e^{-i(t_2-t_1)\comp_1}e^{-i(t_1-t_0)\comp_0}.
\end{equation}
In this paper, the piecewise-constant family is a time-dependent Hamiltonian $\sched$ evaluated at a \emph{rounded clock}, a step function whose value on each interval $(t_j,t_{j+1})$ is a sample time in $[t_j,t_{j+1}]$, so by \eqref{eq:piecewise-product} the product of the exponentials of the samples is the evolution of $\sched$ run on the rounded clock. In \Cref{sec:results}, the clock rounds $s$ between checkpoints $s_j<s<s_{j+1}$ to their midpoint $(s_j+s_{j+1})/2$, and the entries of the samples are truncated too. In \Cref{sec:coarse}, it rounds $t$ to the nearest integer $\lfloor t+1/2\rfloor$, and its two end intervals of length $1/2$ give the half-step phases of \eqref{eq:rounded}.

Steps of length one give the following product.
\begin{definition}[\lean{UnitaryProducts/Comparison/UnitSteps.lean}{86}]\label{def:unit-product}
For a time-dependent Hamiltonian $\sched(s)$ with $s\in[0,1]$ and an integer $T\ge1$, the \emph{$T$-step product} of $\sched$ is
\[
 \devprod_T=e^{-i\sched((T-1)/T)}\cdots e^{-i\sched(0)}.
\]
\end{definition}
Every product of $T$ unitaries is a $T$-step product, since every unitary is $e^{-iA}$ for a Hermitian $A$.
For a time-dependent Hamiltonian $\sched(s)$ with $s\in[0,1]$, $\sched(\lfloor t+1/2\rfloor/T)$ for $t\in[0,T]$ is piecewise constant, equal to $\sched(j/T)$ for $|t-j|<1/2$ and $0\le j\le T$. Its evolution $\compevol$ satisfies the rounded-clock identity
\begin{align}
 \compevol(T)&=e^{-i\sched(1)/2}e^{-i\sched((T-1)/T)}\cdots e^{-i\sched(1/T)}e^{-i\sched(0)/2}\notag\\
 &=e^{-i\sched(1)/2}\,\devprod_T\,e^{i\sched(0)/2}.\label{eq:rounded}
\end{align}
The first equality is \eqref{eq:piecewise-product} for $\sched(\lfloor t+1/2\rfloor/T)$, which equals $\sched(j/T)$ for $t$ in an interval of length $1$ when $0<j<T$ and of length $1/2$ when $j\in\{0,T\}$. The second writes $e^{-i\sched(0)/2}=e^{-i\sched(0)}e^{i\sched(0)/2}$ and uses \Cref{def:unit-product}.

For orthogonal projectors $P$ and $Q$, the \emph{leakage} of a matrix $\devprod$ from $P$ to $Q$ is $\|(1-Q)\devprod P\|$. The next proposition transfers leakage bounds between an evolution and a product in either direction, allowing unitary end phases and perturbed projectors.
\begin{proposition}[Transfer of leakage bounds; \lean{UnitaryProducts/PaperProof/PropTransfer.lean}{56}]\label{prop:transfer}
Let $P$, $Q$, $\tilde{P}$, and $\tilde{Q}$ be orthogonal projectors, let $\evol$ and $\devprod$ be matrices, and let $A$ and $B$ be unitary matrices that commute with $\tilde{Q}$ and $\tilde{P}$, respectively. Then
\[
 \Bigl|\|(1-\tilde{Q})\devprod\tilde{P}\|-\|(1-Q)\evol P\|\Bigr|\le\|A\devprod B-\evol\|+\|\evol\|\bigl(\|\tilde{P}-P\|+\|\tilde{Q}-Q\|\bigr).
\]
\end{proposition}

\begin{proof}
The commutation gives $\|(1-\tilde{Q})A\devprod B\tilde{P}\|=\|(1-\tilde{Q})\devprod\tilde{P}\|$, and projectors are contractions, so the identity
\[
 (1-\tilde{Q})A\devprod B\tilde{P}-(1-Q)\evol P=(1-\tilde{Q})(A\devprod B-\evol)\tilde{P}+(Q-\tilde{Q})\evol P+(1-\tilde{Q})\evol(\tilde{P}-P)
\]
and the reverse triangle inequality prove the proposition.
\end{proof}
Thus a leakage bound for either matrix, such as an adiabatic theorem for an evolution, holds for the other up to the right side.

\subsection{Bands}\label{sec:bands}
The adiabatic theorems transferred in \Cref{sec:results,sec:coarse} concern bands (\Cref{fig:bands}).
\begin{definition}[\lean{UnitaryProducts/Paper/SpectralWindow.lean}{74}]\label{def:spectral-window}
Let $A$ be a Hermitian or a unitary matrix.
\begin{enumerate}
\item If $A$ is Hermitian, a \emph{band} of $A$ is a nonempty proper subset of its eigenvalues of the form $\{\lambda:a<\lambda<b\}$, where $a$ and $b$ are not eigenvalues of $A$. The \emph{spectral gap} $\gap A$ of the band is the least distance $|\lambda-\mu|$ from an eigenvalue $\lambda$ in the band to an eigenvalue $\mu$ outside it.
\item If $A$ is unitary, a \emph{band} of $A$ is a nonempty proper subset of its eigenvalues of the form $\{e^{-i\theta}:a<\theta<b\}$ with $a<b<a+2\pi$, where $e^{-ia}$ and $e^{-ib}$ are not eigenvalues of $A$. The \emph{eigenphase gap} $\gap A$ of the band is the least distance along the circle, $\min_{k\in\mathbb Z}|\theta-\phi+2\pi k|$, from an eigenvalue $e^{-i\theta}$ in the band to an eigenvalue $e^{-i\phi}$ outside it.
\item In either case, the \emph{spectral projector} of a band projects orthogonally onto the span of the eigenvectors of $A$ with eigenvalues in the band.
\item A \emph{band} of a Hermitian or unitary family $A(s)$ with $s\in[0,1]$ is a band of each $A(s)$ with endpoints $a(s)$ and $b(s)$ continuous in $s$. It is a \emph{band of $m$ eigenvalues} if each $A(s)$ has $m$ eigenvalues in it, counted with multiplicity, and its spectral projectors at $s=0$ and $s=1$ are its \emph{endpoint projectors}.
\end{enumerate}
\end{definition}
\begin{figure}[t]
\centering
\begin{tikzpicture}
\fill[gray!20]plot[domain=0:4,samples=40](\x,{0.75+0.1*sin(deg(0.8*\x))})--plot[domain=4:0,samples=40](\x,{1.9+0.05*cos(deg(\x))})--cycle;
\draw[dashed]plot[domain=0:4,samples=40](\x,{0.75+0.1*sin(deg(0.8*\x))})plot[domain=0:4,samples=40](\x,{1.9+0.05*cos(deg(\x))});
\draw[thick]plot[domain=0:4,samples=40](\x,{1.25-0.2*sin(deg(0.8*\x))})plot[domain=0:4,samples=40](\x,{1.15+0.125*\x});
\draw plot[domain=0:4,samples=40](\x,{0.3+0.35*sin(deg(0.8*\x))})plot[domain=0:4,samples=40](\x,{2.3+0.2*cos(deg(\x))})plot[domain=0:4,samples=40](\x,{2.2+0.15*\x});
\foreach \y/\n in {0.28/1,1.26/2,1.65/3,2.17/4,2.8/5}\node[right,font=\scriptsize]at(4.05,\y){$\lambda_\n$};
\node[left,font=\scriptsize]at(-0.05,0.75){$a(s)$};
\node[left,font=\scriptsize]at(-0.05,1.95){$b(s)$};
\node[font=\scriptsize]at(3.25,1.36){band};
\draw[->](0,-0.05)--(4.5,-0.05)node[right]{$s$};
\node[below]at(0,-0.05){$0$};
\node[below]at(4,-0.05){$1$};
\draw[<->](1.96,0.66)--(1.96,1.04)node[midway,right=1pt,fill=white,inner sep=1pt]{$\Gamma$};
\begin{scope}[shift={(7.6,1.3)}]
\draw(0,0)circle(1.35);
\draw[line width=3pt,gray!40](42:1.35)arc(42:102:1.35);
\foreach \a in {42,102}\draw(\a:1.23)--(\a:1.47);
\foreach \a/\n in {15/1,57/2,87/3,160/4,235/5,305/6}{\fill(\a:1.35)circle(1.7pt);\node[font=\scriptsize]at(\a:1.05){$\mu_\n$};}
\draw[<->](15:1.62)arc(15:57:1.62);
\node at(36:1.9){$\Gamma$};
\end{scope}
\end{tikzpicture}
\caption{A band (thick) of a Hermitian family between $a(s)$ and $b(s)$ (dashed), a band (gray arc) of a unitary matrix, and their gaps $\Gamma$. Eigenvalues may cross.}\label{fig:bands}
\end{figure}
\section{Evolutions as products}\label{sec:results}

In this section, we run a regulated time-dependent Hamiltonian on the clock rounded to the midpoints between consecutive checkpoints and truncate the entries of its samples. Checkpoints and truncation can be chosen to make the resulting product's distance from the Hamiltonian's evolution arbitrarily small, even when the Hamiltonian jumps. We then transfer continuous leakage bounds to these products.

\subsection{Level truncation}\label{sec:truncation}
Precision is set by a function $F$ applied to the entries, such as bit truncation (\Cref{ex:truncation-bits}).
\begin{definition}[\lean{UnitaryProducts/Paper/FineStepBounds.lean}{63}]\label{def:level-function}
For a function $F\colon\mathbb R\to\mathbb R$ and a matrix $\sched$, $\sched^{(F)}$ is the matrix with entries
\[
  \sched^{(F)}_{pq}=F(\operatorname{Re}\sched_{pq})+iF(\operatorname{Im}\sched_{pq}),
\]
and for a matrix-valued function $\sched(s)$, $\sched^{(F)}(s)=\sched(s)^{(F)}$.
\end{definition}

Products need Hermitian samples, so we require $F$ to be odd. Then $(A^{(F)})^\dagger=(A^\dagger)^{(F)}$ for every matrix $A$, as conjugating an entry negates its imaginary part and $F(-y)=-F(y)$, so for Hermitian $\sched(s)$ each $\sched^{(F)}(s)$ and each error $\sched(s)-\sched^{(F)}(s)$ is Hermitian. For $d\ge2$, oddness is also necessary: if $\sched^{(F)}$ is Hermitian for every Hermitian $\sched$, the entries $x+iy$ and $x-iy$ of $\sched$ at $(p,q)$ and $(q,p)$ give $F(-y)=-F(y)$ for every real $y$.
Bit truncation makes the entry errors explicit.
\begin{example}[Bits; \lean{UnitaryProducts/PaperProof/ExTruncationBits.lean}{217}]\label{ex:truncation-bits}
Fix an integer $\bits\ge0$ and a Hermitian family $\sched(s)$ with $s\in[0,1]$. Let bit truncation be $[x]_\bits=\operatorname{sgn}(x)\,2^{-\bits}\lfloor2^\bits|x|\rfloor$, let $\sched^{(\bits)}=\sched^{([\cdot]_\bits)}$, and let $\sparsity_\bits=\max_p\max_{s\in[0,1]}\#\{q:\sched^{(\bits)}_{pq}(s)\ne\sched_{pq}(s)\}$ be the largest number of entries that truncation changes in a row of $\sched(s)$.
\begin{enumerate}
\item Bit truncation is odd.
\item For every $s\in[0,1]$, $\|\sched(s)-\sched^{(\bits)}(s)\|\le\sqrt2\,\sparsity_\bits\,2^{-\bits}$.
\end{enumerate}
\end{example}

\begin{proof}
In $[x]_\bits$, only the factor $\operatorname{sgn}(x)$ depends on the sign of $x$, and it is odd, so $[-x]_\bits=-[x]_\bits$, which is (i). For real $x$, as $[x]_\bits$ lies between $0$ and $x$ and $2^\bits|x|-\lfloor2^\bits|x|\rfloor<1$,
\[|x-[x]_\bits|=|x|-|[x]_\bits|<2^{-\bits}.\]
So an entry $z$ of $\sched(s)-\sched^{(\bits)}(s)$ has real and imaginary parts of modulus below $2^{-\bits}$, and
\[
 |z|=\sqrt{(\operatorname{Re}z)^2+(\operatorname{Im}z)^2}<\sqrt{2\cdot4^{-\bits}}=\sqrt2\,2^{-\bits}.
\]
As the error is Hermitian by (i), the Schur test gives
\[
 \|\sched(s)-\sched^{(\bits)}(s)\|\le\max_p\sum_q\bigl|\sched_{pq}(s)-\sched^{(\bits)}_{pq}(s)\bigr|\le\sqrt2\,\sparsity_\bits\,2^{-\bits},
\]
the second inequality by this bound on $|z|$ and the definition of $\sparsity_\bits$. The resulting bound is (ii).
\end{proof}
As $[0]_\bits=0$, truncation keeps zero entries zero, so $\sparsity_\bits$ is at most the largest number of nonzero entries in a row of $\sched(s)$, and it is $0$ whenever $\sched^{(\bits)}=\sched$, however dense $\sched$ is.

\subsection{Midpoint products}\label{sec:main}
To build a product, we round the clock to the midpoint between consecutive checkpoints and truncate the sample there. We choose checkpoints between which the Hamiltonian stays close to that sample.

\begin{definition}[\lean{UnitaryProducts/PaperProof/OscillationFineSteps.lean}{75}]\label{def:oscillation-checkpoints}
Checkpoints $0=s_0\le\cdots\le s_N=1$ have \emph{midpoint error at most $\varepsilon$} for a time-dependent Hamiltonian $\sched$ if $\|\sched(s)-\sched(s_{j+1/2})\|\le\varepsilon$ for every $j$ and every $s$ with $s_j<s<s_{j+1}$, where $s_{j+1/2}=(s_j+s_{j+1})/2$.
\end{definition}
Every regulated time-dependent Hamiltonian admits such checkpoints for every $\varepsilon>0$.
\begin{lemma}[\lean{UnitaryProducts/PaperProof/OscillationFineSteps.lean}{237}]\label{lem:oscillation-partition}
Let $\sched$ be a regulated time-dependent Hamiltonian, and let $\varepsilon>0$. There are checkpoints $0=s_0<\cdots<s_N=1$ with $\|\sched(x)-\sched(y)\|\le\varepsilon$ whenever $s_j<x,y<s_{j+1}$ for some $j$. In particular, their midpoint error is at most $\varepsilon$.
\end{lemma}

\begin{proof}
By \cite[(7.6.1)]{dieudonne1960}, there are checkpoints $0=s_0<\cdots<s_N=1$ and matrices $A_j$ with $\|\sched(s)-A_j\|\le\varepsilon/2$ for $s_j<s<s_{j+1}$. Then $\|\sched(x)-\sched(y)\|\le\varepsilon$ for $s_j<x,y<s_{j+1}$, and $s_{j+1/2}$ is such a $y$.
\end{proof}

Each interval between checkpoints contributes the evolution under the sample at its midpoint for $T$ times the interval's length.
\begin{definition}[\lean{UnitaryProducts/Paper/OneVariable.lean}{65}]\label{def:midpoint-product}
The \emph{midpoint product} of a Hermitian family $A$ over time $T$ on checkpoints $s_0,\ldots,s_N$ is
\[
 e^{-iT(s_N-s_{N-1})A(s_{N-1/2})}\cdots e^{-iT(s_1-s_0)A(s_{1/2})},
\]
where $s_{j+1/2}=(s_j+s_{j+1})/2$.
\end{definition}
For an odd $F$, Duhamel's bound \eqref{eq:duhamel} bounds the distance of the midpoint product of $\sched^{(F)}$ from $\evol(T)$ by the largest distance of $\sched$ from its truncated midpoint sample within an interval, which splits into the change of $\sched$ within the interval, controlled by the checkpoints, and the truncation error of the sample, controlled by $F$.
\begin{theorem}[\lean{UnitaryProducts/PaperProof/OscillationFineSteps.lean}{209}]\label{thm:fine-steps}
Let $F\colon\mathbb R\to\mathbb R$ be odd, and let $\sched$ be a regulated time-dependent Hamiltonian with evolution $\evol$ over time $T>0$. The midpoint product $W$ of $\sched^{(F)}$ over time $T$ on any checkpoints $0=s_0\le\cdots\le s_N=1$ satisfies
\[
 \|W-\evol(T)\|\le T\max_j\sup_{s_j<s<s_{j+1}}\bigl\|\sched(s)-\sched^{(F)}(s_{j+1/2})\bigr\|.
\]
A repeated checkpoint contributes the factor $1$ to $W$ and an empty supremum, taken to be $0$.
\end{theorem}

\begin{proof}
The piecewise-constant family $\comp(s)=\sched^{(F)}(s_{j+1/2})$ for $s_j<s<s_{j+1}$ is Hermitian, as $F$ is odd, and its evolution over time $T$ ends at $W$ by \eqref{eq:piecewise-product}. Duhamel's bound \eqref{eq:duhamel} gives the inequality, as the interval lengths sum to one.
\end{proof}

For $\varepsilon>0$, let $W$ be the midpoint product of $\sched^{(\bits)}$ on checkpoints with midpoint error at most $\varepsilon$, which exist by \Cref{lem:oscillation-partition}. Bit truncation is odd by \Cref{ex:truncation-bits}(i), and for $s_j<s<s_{j+1}$,
\begin{align*}
 \bigl\|\sched(s)-\sched^{(\bits)}(s_{j+1/2})\bigr\|&\le\|\sched(s)-\sched(s_{j+1/2})\|+\bigl\|\sched(s_{j+1/2})-\sched^{(\bits)}(s_{j+1/2})\bigr\|\\
 &\le\varepsilon+\sqrt2\,\sparsity_\bits\,2^{-\bits}
\end{align*}
by \Cref{ex:truncation-bits}(ii), so \Cref{thm:fine-steps} gives
\begin{equation}\label{eq:fine-bits}
 \|W-\evol(T)\|\le T\bigl(\varepsilon+\sqrt2\,\sparsity_\bits\,2^{-\bits}\bigr).
\end{equation}
When $\sparsity_\bits>0$, choosing $\varepsilon=\sqrt2\,\sparsity_\bits\,2^{-\bits}$ makes the right side of \eqref{eq:fine-bits} at most $2\sqrt2\,\sparsity_\bits\,2^{-\bits}T$. Since $\sparsity_\bits\le d$, increasing the precision and decreasing $\varepsilon$ make the error arbitrarily small.

\subsection{Finite-precision transfer}\label{sec:dat}
By \Cref{prop:transfer} with $\tilde P=P$, $\tilde Q=Q$, and $A=B=1$, the leakages $\|(1-Q)WP\|$ and $\|(1-Q)\evol(T)P\|$ of any product $W$ and of $\evol(T)$ differ by at most $\|W-\evol(T)\|$, so leakage bounds transfer both ways.

The next example assumes a bound $\|(1-Q)\evol(T)P\|\le C/T$ on the leakage of the evolution, which an adiabatic theorem supplies under its own hypotheses, and converts it into a bound on the leakage of the product. Jansen, Ruskai, and Seiler \cite[Theorem~3]{jansen2007}, for instance, bound the leakage for a $C^2$ time-dependent Hamiltonian $\sched(s)$ with $s\in[0,1]$, a band of $m$ eigenvalues with spectral gap $\gap{\sched(s)}$ continuous in $s$, and endpoint projectors $P$ at $s=0$ and $Q$ at $s=1$ (\Cref{def:spectral-window}). For the evolution $\evol$ of $\sched$ over time $T$, they prove $\|(1-Q)\evol(T)P\|\le J_m(\sched)/T$ (their $m$ counts distinct eigenvalues, at most ours), where
\begin{equation}\label{eq:jrs-coefficient}
  J_m(\sched)=\sum_{s\in\{0,1\}}\frac{m\|\sched'(s)\|}{\gap{\sched(s)}^2}+\int_0^1\Bigl(\frac{m\|\sched''(s)\|}{\gap{\sched(s)}^2}+7m\sqrt m\,\frac{\|\sched'(s)\|^2}{\gap{\sched(s)}^3}\Bigr)\,ds.
\end{equation}

We choose the checkpoints and the precision for a target leakage.

\begin{example}[Bits; \lean{UnitaryProducts/PaperProof/OscillationFineSteps.lean}{349}]\label{cor:jrs}
Let $\sched$ be a regulated time-dependent Hamiltonian, let $P$ and $Q$ be orthogonal projectors, and let $C>0$. For a target leakage $\delta>0$, let $T=2C/\delta$, and assume that the evolution $\evol$ of $\sched$ over time $T$ satisfies $\|(1-Q)\evol(T)P\|\le C/T$. Choose checkpoints with midpoint error at most $\delta^2/(8C)$ and an integer $\bits\ge0$ with $\sqrt2\,\sparsity_\bits\,2^{-\bits}\le\delta^2/(8C)$. The midpoint product $W$ of $\sched^{(\bits)}$ on these checkpoints over time $T$ satisfies $\|(1-Q)WP\|\le\delta$.
\end{example}

\begin{proof}
\Cref{prop:transfer}, the assumed bound, and \eqref{eq:fine-bits} give
\[
 \|(1-Q)WP\|\le C/T+T\bigl(\delta^2/(8C)+\sqrt2\,\sparsity_\bits\,2^{-\bits}\bigr)
 \le\delta/2+\delta/4+\delta/4=\delta.
\]
\end{proof}
When $\sched$ satisfies the hypotheses of Jansen, Ruskai, and Seiler, their bound gives the hypothesis with $C=J_m(\sched)$ when $J_m(\sched)>0$.

\begin{remark}[\lean{UnitaryProducts/PaperProof/RemTruncatedEndpoints.lean}{159}]\label{rem:truncated-endpoints}
A device that stores $\sched$ only to the $\bits$th binary place uses the spectral projectors of $\sched^{(\bits)}(0)$ and $\sched^{(\bits)}(1)$ rather than those of $\sched(0)$ and $\sched(1)$. For $s\in\{0,1\}$, let $P_s$ be the spectral projector of a band $S_s$ of $\sched(s)$ with spectral gap $\gap{\sched(s)}$ (\Cref{def:spectral-window}), suppose $\varepsilon_s=\|\sched(s)-\sched^{(\bits)}(s)\|<\gap{\sched(s)}/2$, and let $\tilde P_s$ project onto the span of the eigenvectors of $\sched^{(\bits)}(s)$ with eigenvalues within $\varepsilon_s$ of $I_s=[\min S_s,\max S_s]$.

The eigenvalues of $\sched(s)$ outside $S_s$ lie at distance at least $\gap{\sched(s)}$ from $I_s$. By Weyl's perturbation theorem \cite[Corollary~III.2.6]{bhatia1997}, truncation moves each ordered eigenvalue by at most $\varepsilon_s$, so $\tilde P_s$ and $P_s$ have equal rank, and the other eigenvalues of $\sched^{(\bits)}(s)$ lie at distance at least $\gap{\sched(s)}-\varepsilon_s>\varepsilon_s$ from $I_s$. At equal rank $\|P_s-\tilde P_s\|=\|(1-\tilde P_s)P_s\|$ \cite[Exercise~VII.1.11]{bhatia1997}, and the Davis--Kahan $\sin\Theta$ theorem \cite{daviskahan1970}, \cite[Theorem~VII.3.1]{bhatia1997}, bounds $\|(1-\tilde P_s)P_s\|$ by $\varepsilon_s/(\gap{\sched(s)}-\varepsilon_s)$.

\Cref{prop:transfer} with $A=B=1$ and $\|\evol(T)\|=1$ shows that the leakage of $W$ from $\tilde P_0$ to $\tilde P_1$ differs from that of $\evol(T)$ from $P_0$ to $P_1$ by at most $\|W-\evol(T)\|+\sum_{s\in\{0,1\}}\varepsilon_s/(\gap{\sched(s)}-\varepsilon_s)$.
\end{remark}

\section{Products as evolutions}\label{sec:coarse}

This section goes in the direction opposite to \Cref{sec:results}, from a product of $T$ unitaries to the evolution it approximates. Let $\sched(s)$ for $s\in[0,1]$ be a differentiable time-dependent Hamiltonian with one-sided derivatives at $s=0$ and $s=1$ and with Lipschitz derivative, and let $T\ge1$ be an integer. Recall that at unit steps the rounded clock is $\lfloor t+1/2\rfloor$, so that $|t-\lfloor t+1/2\rfloor|\le1/2$. By the rounded-clock identity \eqref{eq:rounded}, the $T$-step product $\devprod_T$ of \Cref{def:unit-product}, with half-step phases at its ends, is the value $\compevol(T)$ of the evolution $\compevol$ generated by $G(t)=\sched(\lfloor t+1/2\rfloor/T)$ for $t\in[0,T]$. We compare $\compevol$ with the evolution $\evol$ generated by $K(t)=\sched(t/T)$. The two generators differ by
\[
 \|K(t)-G(t)\|=\|\sched(t/T)-\sched(\lfloor t+1/2\rfloor/T)\|\le\frac{\|\sched'\|_\infty}T|t-\lfloor t+1/2\rfloor|\le\frac{\|\sched'\|_\infty}{2T}.
\]
This difference is small at every time, but the evolutions run for time $T$, so Duhamel's bound \eqref{eq:duhamel} gives only
\begin{align*}
 \|\compevol(T)-\evol(T)\|
 &\le\int_0^T\|K(t)-G(t)\|\,dt\\
 &\le\frac{\|\sched'\|_\infty}T\int_0^T|t-\lfloor t+1/2\rfloor|\,dt\\
 &=\frac{\|\sched'\|_\infty}4,
\end{align*}
as $\int_k^{k+1}|t-\lfloor t+1/2\rfloor|\,dt=1/4$ for integers $k$. Yet, under the width hypothesis of \Cref{thm:clocks}, the error is of order $1/T$ \eqref{eq:products}. Duhamel's bound loses this order by taking the norm of $K-G$ at each time before integrating. For $t\in[0,T]$ with $|t-j|<1/2$, where $G(t)=\sched(j/T)$, $K(t)-G(t)$ is $(t-j)\sched'(j/T)/T$ to first order, which is odd about $t=j$ and integrates to zero over $|t-j|\le1/2$. The comparison identity of \Cref{thm:comparison} keeps this cancellation through a correction that may jump at $j-1/2$. \Cref{sec:comparison} chooses the correction and proves \Cref{thm:clocks}, which bounds the comparison error up to the correction's jumps, and \Cref{ex:unit-sharp} shows the order $1/T$ attained. \Cref{sec:jumps} bounds the jumps, which gives the comparison \eqref{eq:products}, \Cref{sec:unit-transfer} transfers continuous bounds to the product, and \Cref{sec:walks} treats products of given unitaries.

\subsection{Unit-step products}\label{sec:comparison}
For a correction $Y(t)$ that is continuously differentiable for $|t-j|<1/2$ for each $j$, with one-sided limits at the points $p_j=j-1/2$, \Cref{thm:comparison} writes $\evol(T)^\dagger\compevol(T)-1$ through three kinds of terms: the residual
\[
 R=K-G-Y'-i(KY-YG)
\]
of \eqref{eq:residual}, integrated over $[0,T]$, the boundary values $Y(0)$ and $Y(T)$, and the jumps $\Delta_j$ of $Y$ at $p_j$ for $1\le j\le T$. The residual is integrated over time $T$ and there are $T$ jumps, so a correction with residual and jumps $O(T^{-2})$ and boundary values $O(1/T)$ gives an error $O(1/T)$.

As $G(t)=K(\lfloor t+1/2\rfloor)$, the difference $K(t)-G(t)$ equals $\int_{\lfloor t+1/2\rfloor}^tK'(\sigma)\,d\sigma$, and the residual is
\begin{equation}\label{eq:corrected-residual}
\begin{aligned}
 R&=K-G-Y'-i(KY-YG)\\
 &=K-G-Y'-i[G,Y]-i(K-G)Y\\
 &=\bigl((t-\lfloor t+1/2\rfloor)K'(\lfloor t+1/2\rfloor)-Y'-i[G,Y]\bigr)\\
 &\qquad+\int_{\lfloor t+1/2\rfloor}^t\bigl(K'(\sigma)-K'(\lfloor t+1/2\rfloor)\bigr)\,d\sigma-i(K-G)Y.
\end{aligned}
\end{equation}
The second line writes $KY-YG=[G,Y]+(K-G)Y$, and the third separates $K'(\lfloor t+1/2\rfloor)$ from the integrand. The last two terms are $O(T^{-2})$ when $Y$ is $O(1/T)$, as $K'(\sigma)-K'(\lfloor t+1/2\rfloor)$ is $O(T^{-2})$ and $K-G$ is $O(1/T)$. We therefore choose $Y$ to make the bracket vanish. For $|t-j|<1/2$, $G$ is the constant $\Omega=\sched(j/T)$ and $(t-\lfloor t+1/2\rfloor)K'(\lfloor t+1/2\rfloor)=\tau D$ for $\tau=t-j$ and $D=\sched'(j/T)/T$, so the bracket vanishes when $Y$, as a function of $\tau$, solves
\begin{equation}\label{eq:corrector}
 \begin{cases}
  Y'+i[\Omega,Y]=\tau D,\\
  Y(-\frac12)=Y(\frac12).
 \end{cases}
\end{equation}
The second condition, equal end values, makes the jump at $p_j$ the change of end value from $j-1$ to $j$. That value is $1/T$ times a smooth function of $\Omega$ and $\sched'(j/T)$, which change by $O(1/T)$ between samples, so each jump is $O(T^{-2})$.

Let $\psi_p$ be an orthonormal basis of eigenvectors of $\Omega$, with $\Omega\psi_p=\lambda_p\psi_p$, and write $A_{pq}=\psi_p^\dagger A\psi_q$ for a matrix $A$. As $\Omega$ is Hermitian, $[\Omega,A]_{pq}=\psi_p^\dagger(\Omega A-A\Omega)\psi_q=(\lambda_p-\lambda_q)A_{pq}$, so, with $x=\lambda_p-\lambda_q$, applying $\psi_p^\dagger\cdot\psi_q$ to \eqref{eq:corrector} gives $Y_{pq}'+ixY_{pq}=\tau D_{pq}$. Integrating $(e^{ix\tau}Y_{pq})'=\tau e^{ix\tau}D_{pq}$ over $|\tau|\le1/2$ gives
\begin{equation}\label{eq:cell-integral}
\begin{aligned}
 e^{ix/2}Y_{pq}(\tfrac12)-e^{-ix/2}Y_{pq}(-\tfrac12)
 &=D_{pq}\int_{-1/2}^{1/2}\tau e^{ix\tau}\,d\tau\\
 &=i\Bigl(\frac{2\sin(x/2)}{x^2}-\frac{\cos(x/2)}x\Bigr)D_{pq}\qquad\text{for }x\ne0.
\end{aligned}
\end{equation}
If the end values agree, the left side is $2i\sin(x/2)Y_{pq}(\frac12)$. At $x=2\pi m\ne0$, it vanishes and the right side is $-i(-1)^mD_{pq}/x$, so equal end values force $D_{pq}=0$. We therefore require every sample to have eigenvalue differences of modulus at most a width $\width<2\pi$, the \emph{width hypothesis}. As $\evol(T)$ depends only on $T\sched$, the evolution of $\sched/n$ over time $nT$ also ends at $\evol(T)$, and its $nT$-step product $e^{-i\sched((nT-1)/(nT))/n}\cdots e^{-i\sched(0)/n}$ has samples with every eigenvalue difference divided by $n$, so large $n$ gives the width hypothesis. For $0<|x|\le\width<2\pi$, $\sin(x/2)\ne0$, and dividing by $2i\sin(x/2)$ gives $Y_{pq}(\pm\frac12)=\cormult(x)D_{pq}$ with $\cormult$ as follows.
\begin{definition}[\lean{UnitaryProducts/Paper/UnitComparisonObjects.lean}{31}]\label{def:unit-objects}
For complex $x\notin2\pi\mathbb Z\setminus\{0\}$, let
\begin{equation}\label{eq:N-def}
 \cormult(x)=\begin{cases}
  \bigl(1-(x/2)\cot(x/2)\bigr)/x^2, & x\ne0,\\
  1/12, & x=0.
 \end{cases}
\end{equation}
For a Hermitian $\Omega$ with eigenvalue differences of modulus less than $2\pi$ and a matrix $E$, let $\cormult([\Omega,\cdot])E$ be the matrix with $\psi_p^\dagger\cormult([\Omega,\cdot])E\psi_q=\cormult(\lambda_p-\lambda_q)E_{pq}$.
\end{definition}
As $\cormult(\lambda_p-\lambda_q)$ depends only on the eigenvalues, $\cormult([\Omega,\cdot])E$ is independent of the eigenbasis. At $x=0$, every solution of $Y_{pq}'=\tau D_{pq}$ has equal end values, and we take the limit as $x\to0$ of the solutions for $0<|x|<2\pi$, whose end values are $\cormult(0)D_{pq}=D_{pq}/12$. Under the width hypothesis, \Cref{lem:cell} gives, for $\Omega=\sched(j/T)$ and $D=\sched'(j/T)/T$, a solution $Y$ of \eqref{eq:corrector} with end values $\cormult([\Omega,\cdot])D$. We take this $Y$ for every $j$.

\begin{theorem}[\lean{UnitaryProducts/PaperProof/UniformConsequences.lean}{102}]\label{thm:clocks}
Let $T\ge1$ be an integer, and let $\sched(s)$ for $s\in[0,1]$ be a differentiable time-dependent Hamiltonian with one-sided derivatives at $s=0$ and $s=1$, with
\[
 \|\sched'(s_1)-\sched'(s_2)\|\le L|s_1-s_2|
\]
for $s_1,s_2\in[0,1]$, and with samples $\sched(j/T)$ for $0\le j\le T$ whose eigenvalue differences have modulus at most a width $\width<2\pi$. Then the evolutions $\evol$ of $\sched(t/T)$ and $\compevol$ of $\sched(\lfloor t+1/2\rfloor/T)$ satisfy
\[
 \|\compevol(T)-\evol(T)\|\le\frac1T\Bigl(c_0\bigl(\|\sched'(0)\|+\|\sched'(1)\|\bigr)+\frac{L}{24}+c_1\|\sched'\|_\infty^2\Bigr)+\sum_{j=1}^T\|\Delta_j\|,
\]
where $c_0$ and $c_1$ are the constants of \Cref{app:constants}, and
\[
 \Delta_j=\frac1T\Bigl(\cormult([\sched(j/T),\cdot])\sched'(j/T)-\cormult([\sched((j-1)/T),\cdot])\sched'((j-1)/T)\Bigr).
\]
\end{theorem}

\begin{proof}
Let $Y$ be the correction that \Cref{lem:cell} gives for every $j$, and let $\tau=t-\lfloor t+1/2\rfloor$, so that, for $|t-j|<1/2$, $Y$ as a function of $\tau$ solves \eqref{eq:corrector} with $\Omega=\sched(j/T)$ and $D=\sched'(j/T)/T$, and its jump at $p_j$ is $\Delta_j$. \Cref{thm:comparison} with endpoints $0$ and $T$ gives
\[
 \evol(T)^\dagger\compevol(T)-1=\int_0^Ti\evol^\dagger R\compevol\,dt+i\evol(T)^\dagger Y(T)\compevol(T)-iY(0)-i\sum_{j=1}^T\evol(p_j)^\dagger\Delta_j\compevol(p_j),
\]
so
\begin{equation}\label{eq:unit-comparison}
 \|\compevol(T)-\evol(T)\|=\|\evol(T)^\dagger\compevol(T)-1\|\le\int_0^T\|R(t)\|\,dt+\|Y(0)\|+\|Y(T)\|+\sum_{j=1}^T\|\Delta_j\|.
\end{equation}
\Cref{lem:cell}(ii) at $\tau=0$ bounds $\|Y(0)\|$ and $\|Y(T)\|$ by $c_0\|\sched'(0)\|/T$ and $c_0\|\sched'(1)\|/T$, respectively. For the residual,
\begin{align*}
 \int_0^T\|R(t)\|\,dt
 &=\int_0^T\Bigl\|\int_{\lfloor t+1/2\rfloor}^t\bigl(K'(\sigma)-K'(\lfloor t+1/2\rfloor)\bigr)\,d\sigma-i(K-G)Y\Bigr\|\,dt\\
 &\le\int_0^T\Bigl(\frac{L\tau^2}{2T^2}+\frac{\|\sched'\|_\infty|\tau|}T\Bigl(c_0+\frac{\tau^2}2\Bigr)\frac{\|\sched'\|_\infty}T\Bigr)\,dt\\
 &=\frac1T\int_{-1/2}^{1/2}\Bigl(\frac L2\tau^2+\|\sched'\|_\infty^2|\tau|\Bigl(c_0+\frac{\tau^2}2\Bigr)\Bigr)\,d\tau\\
 &=\frac1T\Bigl(\frac L{24}+c_1\|\sched'\|_\infty^2\Bigr).
\end{align*}
The first line is \eqref{eq:corrected-residual}, whose bracket vanishes as $Y$ solves \eqref{eq:corrector}. The second uses $\|K'(\sigma)-K'(\lfloor t+1/2\rfloor)\|\le L|\sigma-\lfloor t+1/2\rfloor|/T^2$ for the integral term, and $\|K(t)-G(t)\|\le\|\sched'\|_\infty|\tau|/T$ with \Cref{lem:cell}(ii) and $\|D\|\le\|\sched'\|_\infty/T$ for the other. In the third line, $\tau$ runs over $[-1/2,1/2)$ as $t$ runs over each $[k,k+1)$ with $0\le k<T$, and the fourth uses $\int_{-1/2}^{1/2}(\tau^2,|\tau|,|\tau|^3)\,d\tau=(1/12,1/4,1/32)$. Inserting the three bounds into \eqref{eq:unit-comparison} proves the theorem.
\end{proof}

If the values $\sched(s)$ commute, the comparison reduces to quadrature. Then $\evol(T)=\exp(-i\int_0^T\sched(t/T)\,dt)$, the product $\devprod_T$ is the exponential of $-i$ times the left-endpoint sum $\sum_{j=0}^{T-1}\sched(j/T)$, and
\[
 \compevol(T)=e^{-i\sched(1)/2}\devprod_Te^{i\sched(0)/2}=\exp\Bigl(-i\Bigl(\frac12\sched(0)+\sum_{j=1}^{T-1}\sched(j/T)+\frac12\sched(1)\Bigr)\Bigr)
\]
is the exponential of the trapezoidal sum. The half-step phases turn a first-order rule, whose errors of order $1/T$ on each of $T$ steps add up to order $1$, into a second-order rule with total error of order $1/T$.
\begin{example}[\lean{UnitaryProducts/PaperProof/UnitExamples.lean}{188}]\label{ex:unit-sharp}
For an integer $T\ge2$, the time-dependent Hamiltonian $\sched(s)=s^2Z$ with $Z=\operatorname{diag}(1,-1)$ satisfies the hypotheses of \Cref{thm:clocks} with $L=2$ and $\width=2$, and its evolutions $\evol$ and $\compevol$ satisfy the following.
\begin{enumerate}
\item $\|\compevol(T)-\evol(T)\|=2\sin\bigl(1/(12T)\bigr)$, asymptotic to $1/(6T)$.
\item $\|\devprod_T-\evol(T)\|=2\sin\bigl(1/4-1/(12T)\bigr)$, with limit $2\sin(1/4)$.
\end{enumerate}
\end{example}

\begin{proof}
The values commute, so the integral $\int_0^T\sched(t/T)\,dt=TZ/3$, the left-endpoint sum $\sum_{j=0}^{T-1}\sched(j/T)=(T-1)(2T-1)Z/(6T)$, and the trapezoidal sum, which adds $(\sched(1)-\sched(0))/2=Z/2$ to the left-endpoint sum, give
\begin{align*}
 \evol(T)&=e^{-iTZ/3},\\
 \devprod_T&=e^{-i(T/3-1/2+1/(6T))Z},\\
 \compevol(T)&=e^{-i(T/3+1/(6T))Z}.
\end{align*}
For unitaries $A$ and $B$, $\|A-B\|=\|AB^\dagger-1\|$, and $\|e^{i\theta Z}-1\|=2|\sin(\theta/2)|$, as $Z$ has eigenvalues $\pm1$. Hence $\|\compevol(T)-\evol(T)\|=\|e^{-iZ/(6T)}-1\|$ and $\|\devprod_T-\evol(T)\|=\|e^{i(1/2-1/(6T))Z}-1\|$, which are (i) and (ii).
\end{proof}

\subsection{The jumps}\label{sec:jumps}
\Cref{thm:clocks} leaves the jumps $\Delta_j$, the changes of $\cormult([\sched(s),\cdot])\sched'(s)/T$ between adjacent samples. Interpolating $\sched$ linearly between two samples and bounding the change along the interpolation by \Cref{lem:end-value} makes each jump $O(T^{-2})$, which completes the comparison.

\begin{theorem}[\lean{UnitaryProducts/PaperProof/UniformConsequences.lean}{168}]\label{thm:jumps}
Under the hypotheses of \Cref{thm:clocks}, with the $\mathcal E_\width(\sched)$ of \Cref{app:constants}, the evolutions $\evol$ of $\sched(t/T)$ and $\compevol$ of $\sched(\lfloor t+1/2\rfloor/T)$ satisfy
\begin{equation}\label{eq:products}
 \|\compevol(T)-\evol(T)\|\le\frac{\mathcal E_\width(\sched)}T.
\end{equation}
\end{theorem}

\begin{proof}
For $1\le j\le T$, let $l=(j-1)/T$ and $r=j/T$, and let $\Omega(u)=(1-u)\sched(l)+u\sched(r)$ for $u\in[0,1]$, so that $\Omega'=\sched(r)-\sched(l)$ and $\Omega''=0$. Every $\Omega(u)$ has eigenvalue differences of modulus at most $\width$, as the largest eigenvalue is convex in the matrix and the least concave. Hence
\begin{align*}
 T\|\Delta_j\|
 &\le\bigl\|\cormult([\sched(r),\cdot])\bigl(\sched'(r)-T\Omega'\bigr)\bigr\|+\bigl\|\cormult([\sched(l),\cdot])\bigl(T\Omega'-\sched'(l)\bigr)\bigr\|\\
 &\quad+T\bigl\|\cormult([\Omega(1),\cdot])\Omega'-\cormult([\Omega(0),\cdot])\Omega'\bigr\|\\
 &\le\cormult(\width)\bigl(\|\sched'(r)-T\Omega'\|+\|T\Omega'-\sched'(l)\|\bigr)+T\int_0^1\|(\cormult([\Omega,\cdot])\Omega')'\|\,du\\
 &\le\cormult(\width)\Bigl(\frac L{2T}+\frac L{2T}\Bigr)+2T\cormult'(\width)\|\Omega'\|^2\\
 &\le\frac1T\bigl(\cormult(\width)L+2\cormult'(\width)\|\sched'\|_\infty^2\bigr),\\
 T\|\compevol(T)-\evol(T)\|
 &\le c_0\bigl(\|\sched'(0)\|+\|\sched'(1)\|\bigr)+\frac L{24}+c_1\|\sched'\|_\infty^2+T\sum_{j=1}^T\|\Delta_j\|\\
 &\le\mathcal E_\width(\sched).
\end{align*}
The first two lines add and subtract $\cormult([\Omega(1),\cdot])T\Omega'$ and $\cormult([\Omega(0),\cdot])T\Omega'$, as $\Omega(0)=\sched(l)$ and $\Omega(1)=\sched(r)$. The third bounds the first two terms by \eqref{eq:series-bound} and the last by the integral of the norm of its derivative. The fourth uses $\sched'(r)-T\Omega'=T\int_l^r(\sched'(r)-\sched'(\sigma))\,d\sigma$, of norm at most $LT\int_l^r(r-\sigma)\,d\sigma=L/(2T)$, the same bound at $l$, and \Cref{lem:end-value}(i) with $\Omega''=0$. The fifth uses $\|\Omega'\|\le\|\sched'\|_\infty/T$. The sixth is \Cref{thm:clocks}, and the last sums the fifth over $1\le j\le T$ and collects the coefficients of $\mathcal E_\width(\sched)$.
\end{proof}

\Cref{introthm:unit} is \eqref{eq:products} at width $\pi$ (\Cref{fact:pi}).

\subsection{Transfer of continuous bounds}\label{sec:unit-transfer}
If orthogonal projectors $P$ and $Q$ commute with $\sched(0)$ and $\sched(1)$, respectively, they commute with the half-step phases, so \Cref{prop:transfer} transfers leakage bounds between $\devprod_T$ and $\evol(T)$, adding the comparison error \eqref{eq:products} and, for truncated samples, their truncation error (\Cref{ex:unit-transfer}).

\begin{example}[Bits; \lean{UnitaryProducts/PaperProof/UniformConsequences.lean}{851}]\label{ex:unit-transfer}
Let $\sched$, $T$, $L$, and $\width$ satisfy the hypotheses of \Cref{thm:clocks}, with the $\mathcal E_\width(\sched)$ of \Cref{app:constants}, let $P$ and $Q$ be orthogonal projectors commuting with $\sched(0)$ and $\sched(1)$, respectively, and let $C\ge0$. For a target leakage $\delta>0$, assume that $\|(1-Q)\evol(T)P\|\le C/T$ and $T\ge2(C+\mathcal E_\width(\sched))/\delta$, and choose an integer $\bits\ge0$ with $\sqrt2\,\sparsity_\bits\,2^{-\bits}\le\delta/(2T)$. The $T$-step product $\devprod^{(\bits)}_T$ of the truncation $\sched^{(\bits)}$ of \Cref{ex:truncation-bits} satisfies $\|(1-Q)\devprod^{(\bits)}_TP\|\le\delta$.
\end{example}

\begin{proof}
\Cref{prop:transfer} with $A=e^{-i\sched(1)/2}$ and $B=e^{i\sched(0)/2}$, which commute with $Q$ and $P$, with $\tilde P=P$ and $\tilde Q=Q$, and with $\|(1-Q)\evol(T)P\|\le C/T$ gives the first line of
\begin{align*}
 \|(1-Q)\devprod^{(\bits)}_TP\|
 &\le\frac CT+\bigl\|e^{-i\sched(1)/2}\devprod^{(\bits)}_Te^{i\sched(0)/2}-\evol(T)\bigr\|\\
 &\le\frac CT+\|\devprod^{(\bits)}_T-\devprod_T\|+\|\compevol(T)-\evol(T)\|\\
 &\le\frac{C+\mathcal E_\width(\sched)}T+T\sqrt2\,\sparsity_\bits\,2^{-\bits}\\
 &\le\frac\delta2+\frac\delta2=\delta.
\end{align*}
The second line is \eqref{eq:rounded}. In the third, the bit truncation $[\cdot]_\bits$ is odd by \Cref{ex:truncation-bits}, so $\sched^{(\bits)}(\lfloor t\rfloor/T)$ is Hermitian and generates $\devprod^{(\bits)}_T$ by \eqref{eq:piecewise-product}, and Duhamel's bound \eqref{eq:duhamel} with \Cref{ex:truncation-bits} bounds $\|\devprod^{(\bits)}_T-\devprod_T\|$ by $T\sqrt2\,\sparsity_\bits\,2^{-\bits}$, while \eqref{eq:products} bounds $\|\compevol(T)-\evol(T)\|$ by $\mathcal E_\width(\sched)/T$. The last uses the choices of $T$ and $\bits$.
\end{proof}

\subsection{Products of given unitaries}\label{sec:walks}
So far the factors of the product have been exponentials of samples of a time-dependent Hamiltonian. Discrete adiabatic theorems take the factors instead from a unitary family $\devprod(s)$ for $s\in[0,1]$, the walks of \cite{costa2022,an2025}, often $C^2$ with bounds on $\devprod'$ and $\devprod''$, under eigenphase gaps. For a continuous Hermitian $\sched(s)$ with spectrum in $[a(s),a(s)+\width]$ for a continuous function $a$, the unitary $\devprod=e^{-i\sched}$ has the eigenvalues of $\sched$ read modulo $2\pi$, so it satisfies
\begin{equation}\label{eq:arc}
 \operatorname{spec}\devprod(s)\subset\{e^{-i\theta}:a(s)\le\theta\le a(s)+\width\}\quad\text{for every $s\in[0,1]$},
\end{equation}
and an open arc of length $\freearc=2\pi-\width$ contains no eigenvalue (\Cref{fig:free-arc}). The least eigenvalue of $\sched(s)$ is then one such $a$, as it is continuous \cite[Corollary~III.2.6]{bhatia1997}. The condition that $\sched(s)$ has spectrum in $[a(s),a(s)+\width]$ and \eqref{eq:arc} differ only in the branch of the logarithm, the interval $[a(s),a(s)+\width]$, and conversely, a $C^2$ unitary family that satisfies \eqref{eq:arc} for a continuous $a$ is $e^{-i\sched}$ for such a $C^2$ $\sched$ (\Cref{lem:walk-log}). A fixed arc about $-1$ is the case $\|\sched(s)\|\le\width/2$ with $\sched=i\operatorname{Log}\devprod$, and the gaps between consecutive eigenphases of $\devprod(s)$ are those of $\operatorname{spec}\sched(s)+2\pi\mathbb Z$.
\begin{figure}[t]
\centering
\begin{tikzpicture}[>=stealth]
 \draw (0,0) circle (1.3);
 \draw[line width=3pt,gray!40] (-35:1.3) arc (-35:215:1.3);
 \draw[very thick,dashed] (215:1.3) arc (215:325:1.3);
 \foreach \a in {-20,30,75,120,160,200} \fill (\a:1.3) circle (1.7pt);
 \node[anchor=north east] at (215:1.3) {$e^{-ia}$};
 \node[anchor=north west] at (-35:1.3) {$e^{-i(a+\width)}$};
 \node at (270:1.65) {no eigenvalue, length $\freearc$};
 \node at (90:1.7) {spectrum, in an arc of length $\width=2\pi-\freearc$};
\end{tikzpicture}
\caption{The width $\width$ of $\sched(s)$ on the circle, with its spectrum in $[a,a+\width]$: the eigenvalues of $e^{-i\sched(s)}$ (dots) lie in the closed gray arc, and the open dashed arc contains none.}\label{fig:free-arc}
\end{figure}

\begin{theorem}[\lean{UnitaryProducts/PaperProof/CorProducts.lean}{354}]\label{cor:walks}
Let $T$, $\sched$, $L$, and $\width$ satisfy the hypotheses of \Cref{thm:clocks}, with $\sched$ twice continuously differentiable and every $\sched(s)$ with eigenvalue differences of modulus at most $\width$, and let $\devprod=e^{-i\sched}$. Then the evolution $\evol$ of $\sched(t/T)$ and $\compevol(T)=e^{-i\sched(1)/2}\devprod((T-1)/T)\cdots\devprod(0)e^{i\sched(0)/2}$ satisfy
\begin{equation}\label{eq:walks}
 \|\compevol(T)-\evol(T)\|\le\frac{\widetilde{\mathcal E}_\width(\devprod)}T
\end{equation}
for the constant $\widetilde{\mathcal E}_\width(\devprod)$ of \Cref{app:constants}, whose coefficients are $O((2\pi-\width)^{-3})$ as $\width\to2\pi$ (\Cref{fact:arc}).
\end{theorem}

\begin{proof}
By \eqref{eq:rounded}, $\compevol(T)$ is the evolution of $\sched(\lfloor t+1/2\rfloor/T)$ at $T$. As $\devprod$ is $C^2$ and satisfies \eqref{eq:arc} with $a$ the least eigenvalue of $\sched$, $\sched$ is the logarithm of \Cref{lem:walk-log}, so $\|\sched'\|\le\kappa\|\devprod'\|$ and $\|\sched''\|\le\kappa\|\devprod''\|+\|\devprod'\|^2/\cos^2(\width/4)$. For $1\le j\le T$, with $l=(j-1)/T$ and $r=j/T$,
\begin{align*}
 T\|\Delta_j\|&=\Bigl\|\int_l^r\bigl(\cormult([\sched,\cdot])\sched'\bigr)'\,ds\Bigr\|\\
 &\le\frac1T\bigl(\cormult(\width)\|\sched''\|_\infty+2\cormult'(\width)\|\devprod'\|_\infty\|\sched'\|_\infty+(\cormult(\width)-\cormult(0))\|\sched'\|_\infty^2\bigr),\\
 T\|\compevol(T)-\evol(T)\|
 &\le c_0\bigl(\|\sched'(0)\|+\|\sched'(1)\|\bigr)+\frac{\|\sched''\|_\infty}{24}+c_1\|\sched'\|_\infty^2+T\sum_{j=1}^T\|\Delta_j\|\\
 &\le c_0\bigl(\|\sched'(0)\|+\|\sched'(1)\|\bigr)+\Bigl(\frac1{24}+\cormult(\width)\Bigr)\|\sched''\|_\infty\\
 &\quad+\bigl(c_1+\cormult(\width)-\cormult(0)\bigr)\|\sched'\|_\infty^2+2\cormult'(\width)\|\devprod'\|_\infty\|\sched'\|_\infty\\
 &\le\widetilde{\mathcal E}_\width(\devprod).
\end{align*}
The first line is the definition of $\Delta_j$ in \Cref{thm:clocks}. The second is \Cref{lem:end-value}(ii) for $\Omega=\sched$, with $e^{-i\sched}=\devprod$, bounded by its suprema over an interval of length $1/T$. The third is \Cref{thm:clocks} with $L=\|\sched''\|_\infty$, the next two sum the second over $1\le j\le T$, and the last inserts the two bounds of \Cref{lem:walk-log} and collects the coefficient $c_2$ of \Cref{app:constants}.
\end{proof}

For a unitary family with a gapped band, taking the logarithm with its branch cut at the band's lower endpoint turns \Cref{cor:walks} into a leakage bound for the product of the given unitaries.

\begin{theorem}[\lean{UnitaryProducts/PaperProof/CorWalkLeakage.lean}{157}]\label{cor:walk-leakage}
Let $T\ge1$ be an integer, let a $C^2$ unitary family $\devprod$ have a band of $m$ eigenvalues $\{e^{-i\theta}:\alpha(s)<\theta<\beta(s)\}$ with $\alpha<\beta<\alpha+2\pi$, eigenphase gap $\gap{\devprod(s)}$ continuous in $s$, and endpoint projectors $P$ and $Q$, and let $\freearc$ be the least $\gap{\devprod(s)}$ over $s\in[0,1]$. Let $\evol$ be the evolution of $\sched(t/T)$ for the Hermitian $\sched(s)$ with $e^{-i\sched(s)}=\devprod(s)$ and spectrum in $(\alpha(s),\alpha(s)+2\pi)$. Then, with $\width=2\pi-\freearc$ and the $\widetilde{\mathcal E}_\width(\devprod)$ of \Cref{app:constants},
\[
 \bigl\|(1-Q)\devprod((T-1)/T)\cdots\devprod(0)P\bigr\|\le\|(1-Q)\evol(T)P\|+\frac{\widetilde{\mathcal E}_\width(\devprod)}T.
\]
\end{theorem}

\begin{proof}
As $\alpha$ and $\devprod$ are continuous, $\sched=\alpha+\pi+i\operatorname{Log}(e^{i(\alpha+\pi)}\devprod)$ is continuous, for the principal logarithm $\operatorname{Log}$, and therefore so is its least eigenvalue $\ell$~\cite[Corollary~III.2.6]{bhatia1997}. As $\sched(s)$ has spectrum in $(\alpha(s),\alpha(s)+2\pi)$, an eigenvalue $\theta$ of $\sched(s)$ lies in $(\alpha(s),\beta(s))$ exactly when $e^{-i\theta}$ lies in the band, and otherwise in $(\beta(s),\alpha(s)+2\pi)$. The band is nonempty and proper, so $\ell(s)$ lies in $(\alpha(s),\beta(s))$ and the greatest eigenvalue $\varphi$ of $\sched(s)$ in $(\beta(s),\alpha(s)+2\pi)$, and \Cref{def:spectral-window}(ii) gives $\freearc\le2\pi-(\varphi-\ell(s))$. Hence $\sched(s)$ has spectrum in $[\ell(s),\ell(s)+\width]$, and $\sched$ is $C^2$ by \Cref{lem:walk-log}. So \Cref{cor:walks} applies, and $\devprod((T-1)/T)\cdots\devprod(0)=\devprod_T$. As $\devprod(s)=e^{-i\sched(s)}$, $P$ and $Q$ are the spectral projectors of $\sched(0)$ and $\sched(1)$ for their eigenvalues in $(\alpha(0),\beta(0))$ and $(\alpha(1),\beta(1))$, and they commute with $e^{i\sched(0)/2}$ and $e^{-i\sched(1)/2}$. \Cref{prop:transfer} with these phases as $B$ and $A$, with $\tilde P=P$ and $\tilde Q=Q$, gives $\|(1-Q)\devprod_TP\|\le\|(1-Q)\evol(T)P\|+\|\compevol(T)-\evol(T)\|$ by \eqref{eq:rounded}, and \eqref{eq:walks} bounds the last term.
\end{proof}

The eigenvalues of $\sched(s)$ in $(\alpha(s),\beta(s))$ form a band of $m$ eigenvalues with the spectral projectors of the band of $\devprod$ and spectral gap at least the eigenphase gap $\gap{\devprod(s)}$ by \Cref{def:spectral-window}(i) and~(ii), continuous in $s$ as the eigenvalues of $\sched$, $\alpha$, and $\beta$ are. The Jansen--Ruskai--Seiler bound of \Cref{sec:dat} therefore bounds $\|(1-Q)\evol(T)P\|$ in \Cref{cor:walk-leakage} by $J_m(\sched)/T$, which gives a discrete adiabatic theorem for walks, with the gap entering through $J_m(\sched)$ and through the coefficients of $\widetilde{\mathcal E}_\width(\devprod)$, which are $O(\freearc^{-3})$ for $\width=2\pi-\freearc$ (\Cref{fact:arc}).
\section{Examples}\label{sec:examples}

I apply the comparison to schedules adapted to the spectral gap, to a precessing spin, and to eigenvalue crossings. Throughout, $T\ge1$ is an integer, $X$ and $Z=\operatorname{diag}(1,-1)$ are Pauli matrices, $\evol$ is the evolution of a time-dependent Hamiltonian $\sched$ over time $T$, $\devprod_T$ is its $T$-step product, and $\compevol(T)=e^{-i\sched(1)/2}\devprod_Te^{i\sched(0)/2}$, as in \eqref{eq:rounded}.

Guo and An~\cite{guo2025} prove, for general paths under a condition on how much of the path has a small gap, that a power-law schedule adapted to the spectral gap prepares a ground state with error inverse-linear in the gap, as Roland and Cerf's local schedule does for Grover search \cite{roland2002} and An and Lin's schedules do for linear systems \cite[Theorem~1]{anlin2022}. For a linear schedule, Guo and An's bound is inverse-quadratic \cite[Theorem~3]{guo2025}. They leave open whether a quantum circuit keeps the inverse-linear dependence \cite[Section~5]{guo2025}. For products of the exponentials of samples, whose factors are evolutions under fixed Hamiltonians that a circuit implements one at a time, \Cref{ex:gap-adapted} answers that question as an immediate consequence of \Cref{ex:unit-transfer}. Earlier discrete results with inverse-linear dependence treat Grover search \cite[Theorem~15]{an2025}, for one of its schedules up to a factor $\log N$, and linear systems \cite{costa2022}, or assume bounds on the discrete variation and a multistep gap \cite[Theorem~15]{costa2026}.

\begin{example}[Gap-adapted schedules; \lean{UnitaryProducts/PaperProof/GapAdapted.lean}{469}]\label{ex:gap-adapted}
Let $\sched_0$ and $\sched_1$ be Hermitian matrices of norm at most $1$, and for $x\in[0,1]$ let $\Gamma(x)$ be the spectral gap of the least eigenvalue of $(1-x)\sched_0+x\sched_1$, with $\Gamma_*=\inf_{x\in[0,1]}\Gamma(x)>0$. Let $1<p<2$ and $c_p=\int_0^1\Gamma(x)^{-p}\,dx$, let $u(s)$ for $s\in[0,1]$ solve $u'=c_p\Gamma(u)^p$ with $u(0)=0$, and let $\sched=(1-u)\sched_0+u\sched_1$. Then $\sched$ satisfies the hypotheses of \Cref{thm:clocks} at width $2$, with
\begin{align*}
 \|\sched'(0)\|+\|\sched'(1)\|&=c_p\|\sched_1-\sched_0\|\bigl(\Gamma(0)^p+\Gamma(1)^p\bigr),\\
 \|\sched'\|_\infty&\le2^pc_p\|\sched_1-\sched_0\|,\\
 L&\le p4^pc_p^2\|\sched_1-\sched_0\|^2,
\end{align*}
so \Cref{ex:unit-transfer} applies to every bound $\|(1-Q)\evol(T)P\|\le C/T$ for the spectral projectors $P$ and $Q$ of $\sched_0$ and $\sched_1$ for their least eigenvalues.
\end{example}

\begin{proof}
By Weyl's perturbation theorem \cite[Corollary~III.2.6]{bhatia1997}, the ordered eigenvalues of $(1-x)\sched_0+x\sched_1$ lie in $[-1,1]$ and are Lipschitz in $x$ with constant $\|\sched_1-\sched_0\|$. As $\Gamma_*>0$, the least eigenvalue keeps its multiplicity, so $\Gamma\le2$ is Lipschitz with constant $2\|\sched_1-\sched_0\|$, and every sample has eigenvalue differences of modulus at most $2$. Integrating $\Gamma(u)^{-p}u'=c_p$ gives $u(1)=1$, so $\sched(0)=\sched_0$ and $\sched(1)=\sched_1$. The displayed bounds follow from $\sched'=c_p\Gamma(u)^p(\sched_1-\sched_0)$, the second as $\Gamma\le2$, and the third by multiplying the Lipschitz constants $p2^{p-1}$ of $y\mapsto y^p$ for $0<y\le2$, $2\|\sched_1-\sched_0\|$ of $\Gamma$, and $2^pc_p$ of $u$.
\end{proof}

The continuous bound $\|(1-Q)\evol(T)P\|\le C/T$ to which \Cref{ex:gap-adapted} applies is the adiabatic theorem of Guo and An. If $u$ is twice continuously differentiable, as their proof uses, and $\{x\in[0,1]:\Gamma(x)\le y\}$ has length at most $Ky$ for every $y>0$, they bound the infidelity $1-\|Q\evol(T)\psi\|^2$ of each unit vector $\psi=P\psi$ by $\eta^2$, where $\eta=O(1/(T\Gamma_*))$ over paths with fixed $p$ and $K$ \cite[Lemma~1, (18), and Theorem~5]{guo2025}, so $\|(1-Q)\evol(T)P\|\le\eta$, which is $C/T$ with $C$ of order $1/\Gamma_*$. The same condition gives $c_p=\int_0^{\Gamma_*^{-p}}|\{x\in[0,1]:\Gamma(x)<y^{-1/p}\}|\,dy\le pK\Gamma_*^{1-p}/(p-1)$ \cite[(19)]{guo2025}, so the $\mathcal E_2(\sched)$ of \Cref{app:constants} is $O(\Gamma_*^{2-2p})$, which for $1<p\le3/2$ is $O(1/\Gamma_*)$, the order of $C$, as $\Gamma_*\le2$. Hence, for $1<p\le3/2$, $O(1/(\delta\Gamma_*))$ factors, even with entries truncated as in \Cref{ex:unit-transfer}, prepare the ground space of $\sched_1$ to within $\delta$, as the continuous schedule does in time $O(1/(\delta\Gamma_*))$.

A spin in a field of fixed strength whose direction precesses is the simplest system whose Hamiltonians at different times do not commute, and holding the field fixed over each time step keeps it solvable.

\begin{example}[A precessing spin; \lean{UnitaryProducts/PaperProof/PrecessingSpin.lean}{1247}]\label{ex:precessing-spin}
Let $0<\omega\le2\pi$ and $\theta$ be real, and let $\sched(s)=\frac\omega2e^{-i\theta sZ/2}Xe^{i\theta sZ/2}$ for $s\in[0,1]$.
\begin{enumerate}
\item If $\omega$ is $\pi$, then, as $T\to\infty$ through odd integers, $T\|\compevol(T)-\evol(T)\|$ tends to $c_0\bigl(\|\sched'(0)\|+\|\sched'(1)\|\bigr)\bigl(1+\theta^2/(4(\pi-2)^2)\bigr)^{1/2}$, with $c_0=1/(2\pi)-1/\pi^2$, the constant of \Cref{thm:clocks} at width $\pi$.
\item If $\omega$ is $2\pi$, then $\|\compevol(T)-\evol(T)\|$ tends to $2|\sin(\theta/4)|$ as $T\to\infty$.
\end{enumerate}
\end{example}

\begin{proof}
In the frame that rotates with the field, the evolution and the product with its end phases are
\begin{align*}
 \evol(T)&=e^{-i\theta Z/2}e^{-i(\omega TX-\theta Z)/2},\\
 \compevol(T)&=e^{-i\theta Z/2}S^T,\qquad S=e^{-i\omega X/4}e^{i\theta Z/(2T)}e^{-i\omega X/4}.
\end{align*}
The first holds as $\sched(t/T)=e^{-i\theta tZ/(2T)}\frac\omega2Xe^{i\theta tZ/(2T)}$, and the second follows from \eqref{eq:rounded}, as the factors are $e^{-i\theta jZ/(2T)}e^{-i\omega X/2}e^{i\theta jZ/(2T)}$.

For (i), with $\omega$ equal to $\pi$ and $T$ odd, let
\begin{align*}
 N&=\cos\tfrac\theta{2T}X-\sin\tfrac\theta{2T}Z, & \rho&=\tfrac12\bigl(\pi^2T^2+\theta^2\bigr)^{1/2},\\
 M&=\frac{\pi TX-\theta Z}{2\rho}, & \zeta&=\rho-\frac{\pi T}2.
\end{align*}
Then
\begin{align*}
 \|\compevol(T)-\evol(T)\|&=\bigl\|(-iN)^T-(-iM)^Te^{-i\zeta M}\bigr\|\\
 &=\|N-M\cos\zeta+i\sin\zeta\|\\
 &=\bigl(\|N-M\cos\zeta\|^2+\sin^2\zeta\bigr)^{1/2},\\
 T\|N-M\cos\zeta\|&\to\frac{(\pi-2)|\theta|}{2\pi}=c_0\bigl(\|\sched'(0)\|+\|\sched'(1)\|\bigr),\\
 T\zeta&\to\frac{\theta^2}{4\pi}.
\end{align*}
The first line uses $S=-iN$, as $e^{-i\pi X/4}=(1-iX)/\sqrt2$, and $e^{-i(\pi TX-\theta Z)/2}=e^{-i\rho M}=(-iM)^Te^{-i\zeta M}$. The second uses $N^T=N$, $M^T=M$, and $Me^{-i\zeta M}=M\cos\zeta-i\sin\zeta$. The third uses $(A+is)^\dagger(A+is)=A^2+s^2$ for real $s$, with $A^2=\|A\|^2$ as $A=N-M\cos\zeta$ is a real combination of $X$ and $Z$. The fourth holds as the axes of $N$ and $M$ make the angle $\frac\theta{2T}-\arctan\frac\theta{\pi T}$ and $\|\sched'\|=\pi|\theta|/2$, and the fifth as $\zeta=\theta^2/(2(2\rho+\pi T))$. Together they give (i).

For (ii), with $\omega$ equal to $2\pi$, every factor is $-1$, so $\compevol(T)=(-1)^Te^{-i\theta Z}$, and, as $e^{-i(2\pi TX-\theta Z)/2}-(-1)^T\to0$,
\[
 \|\compevol(T)-\evol(T)\|=\bigl\|(-1)^Te^{-i\theta Z/2}-e^{-i(2\pi TX-\theta Z)/2}\bigr\|\to\|e^{-i\theta Z/2}-1\|=2|\sin(\theta/4)|.
\]
\end{proof}

The spin shows that the endpoint coefficient $c_0$ cannot be lowered and that the width hypothesis cannot be relaxed. In (i), the terms of $\mathcal E_\pi(\sched)$ other than $c_0(\|\sched'(0)\|+\|\sched'(1)\|)$ are $O(\theta^2)$, so letting $\theta\to0$ shows that no constant smaller than $c_0$ can replace it in \Cref{thm:clocks} at width $\pi$. In (ii), with $\theta=2\pi$, the distance tends to $2$, so the width hypothesis cannot be relaxed to $\width\le2\pi$.

The same computation gives the optimality of $c_0$ at every width. For $0<\omega<2\pi$, multiplying out $S$ gives
\begin{align*}
 S&=\cos\tfrac\theta{2T}\cos\tfrac\omega2-i\Bigl(\cos\tfrac\theta{2T}\sin\tfrac\omega2\,X-\sin\tfrac\theta{2T}\,Z\Bigr)=e^{-i\xi N},\\
 e^{-i(\omega TX-\theta Z)/2}&=e^{-i\rho M},\qquad M=\frac{\omega TX-\theta Z}{2\rho},\qquad\rho=\tfrac12\bigl(\omega^2T^2+\theta^2\bigr)^{1/2},
\end{align*}
for the unit combination $N$ of $X$ and $Z$ proportional to the bracket and an angle $\xi$ within $O(\theta^2/T^2)$ of $\omega/2$, while $\rho$ is within $O(\theta^2/T)$ of $\omega T/2$. The axes of $N$ and $M$ make the angle $\arctan\bigl(\tan\frac\theta{2T}/\sin\frac\omega2\bigr)-\arctan\frac\theta{\omega T}$. When $T\omega/\pi$ is odd, $e^{-i\omega TN/2}$ and $e^{-i\omega TM/2}$ are $-iN$ and $-iM$ up to one sign, so, with $c_0$ at width $\omega$ and $O(\theta^2)$ uniform in $T$,
\begin{align*}
 T\|\compevol(T)-\evol(T)\|&=T\|N-M\|+O(\theta^2),\\
 T\|N-M\|&\to|\theta|\Bigl(\frac1{2\sin(\omega/2)}-\frac1\omega\Bigr)=c_0\bigl(\|\sched'(0)\|+\|\sched'(1)\|\bigr)\qquad\text{as }T\to\infty.
\end{align*}
Such $T$ exist whenever $\omega/\pi$ is a ratio of odd integers, as every odd multiple of the denominator makes $T\omega/\pi$ odd, and these $\omega$ are dense in $(0,2\pi)$. Fix $\width$ with $0<\width<2\pi$, and suppose that a constant $c$ can replace $c_0$ in \Cref{thm:clocks} at width $\width$, with the other coefficients of its bound held fixed. For each such $\omega\le\width$, the spin satisfies the hypotheses of \Cref{thm:clocks} at width $\width$, with the other terms of its bound $O(\theta^2/T)$, so letting $T\to\infty$ through such $T$, and then $\theta\to0$, shows that $c$ is at least $c_0$ at width $\omega$. Letting such $\omega$ increase to $\width$ shows that $c$ is at least $c_0$ at width $\width$, as $c_0$ is continuous in the width. So $c_0$ is optimal at every width $\width$ with $0<\width<2\pi$.

When eigenvalues of $\sched$ cross, the eigenprojector of a level can continue through the crossing, and Avron and Elgart prove an adiabatic theorem for it without a gap \cite[Theorem~4]{avronelgart1999}. The comparison carries their theorem to the product.

\begin{example}[Eigenvalue crossings; \lean{UnitaryProducts/PaperProof/Crossings.lean}{131}]\label{ex:crossings}
Let $\sched$, $L$, and $\width$ satisfy the hypotheses of \Cref{thm:clocks} for every integer $T\ge1$, which holds when every $\sched(s)$ has eigenvalue differences of modulus at most $\width$, and let $P(s)$ for $s\in[0,1]$ be continuous and, for every $s$ other than one point, a spectral projector of $\sched(s)$ for one eigenvalue. If $\|(1-P(1))\evol(T)P(0)\|\to0$ as $T\to\infty$, as Avron and Elgart prove when $P$ is also piecewise twice differentiable, then $\|(1-P(1))\devprod_TP(0)\|\to0$.
\end{example}

\begin{proof}
By continuity, $P(0)$ and $P(1)$ commute with $\sched(0)$ and $\sched(1)$, so \Cref{prop:transfer} with $A=e^{-i\sched(1)/2}$ and $B=e^{i\sched(0)/2}$, as in the proof of \Cref{ex:unit-transfer}, bounds $\|(1-P(1))\devprod_TP(0)\|$ by $\|(1-P(1))\evol(T)P(0)\|+\|\compevol(T)-\evol(T)\|$, and the second term tends to zero by \eqref{eq:products}.
\end{proof}

This is a discrete adiabatic theorem through an eigenvalue crossing. The discrete adiabatic theorems of \Cref{sec:intro} all assume an eigenphase gap, and Tanaka \cite{tanaka2011} excludes crossings.
\section{Discussion}\label{sec:outlook}

A product of the exponentials of samples of $\sched$ is the evolution of $\sched$ run on a rounded clock, and both directions of this paper come from taking that statement literally. The product and the evolution it approximates are then two evolutions, and Duhamel's formula integrated by parts against a correction $Y$ that may jump (\Cref{thm:comparison}) compares them in either direction. I chose to compare two evolutions, rather than an evolution with adiabatic transport, because that choice keeps the spectral gap of $\sched$ out of the comparison. Nothing in the comparisons assumes that $\sched$ varies slowly, so adiabatic theorems are one application, and the discrete adiabatic theorems of \Cref{sec:unit-transfer} are continuous ones carried over by the comparison.

The two directions differ only in the correction $Y$ that \Cref{thm:comparison} integrates against. With $Y=0$, the comparison is Duhamel's bound, which integrates the norm of the generator difference over time. That choice suits \Cref{sec:results}, as it asks only for one-sided limits of $\sched$, so it covers time-dependent Hamiltonians that jump, and it splits the error into a part the checkpoints control and a part the precision controls \eqref{eq:fine-bits}. At unit steps it loses an order, which \Cref{sec:comparison} recovers by choosing $Y$ to solve $Y'+i[\sched(j/T),Y]=(t-j)\sched'(j/T)/T$ with equal values at $j\pm\frac12$ \eqref{eq:corrector}.

The rest of the unit-step comparison follows from centering the clock on the samples, which requires the two half-step phases. Where the values of $\sched$ commute, the phases turn the left-endpoint sum into the trapezoidal sum, and in \Cref{ex:unit-sharp} the distance from the evolution is $2\sin(1/(12T))$ with them and tends to $2\sin(1/4)$ without them. They commute with the spectral projectors of $\sched(0)$ and $\sched(1)$, so a leakage bound for the product with its end phases is one for the product itself (\Cref{prop:transfer}).

The width hypothesis enters only through solving \eqref{eq:corrector}. It excludes no evolution and sets only how many factors a given evolution needs (\Cref{sec:comparison}). A product of given unitaries, such as the walk of a discrete adiabatic theorem, is the $T$-step product of any Hermitian logarithm of the family they come from, and a family with a band of eigenphase gap $\freearc$ has a logarithm of width $2\pi-\freearc$ (\Cref{cor:walk-leakage}). The large steps of An, Costa, and Berry \cite{an2025} are thus the unit steps of \Cref{sec:coarse}, and \Cref{cor:walk-leakage} recovers their $O(1/T)$ rate in $T$ from the continuous theorem, keeping its constant as given, through a comparison that holds for every initial state and whose coefficients are $O(\freearc^{-3})$ in the eigenphase gap $\freearc$.

\subsection{Open questions}
The comparison leaves five questions open.
\begin{enumerate}
\item \emph{Higher order.} Some continuous theorems decay faster than $1/T$, such as the leakage bounds $O(\exp(-\gamma T^\rho))$, with $\gamma>0$ and $0<\rho<1$, for gapped Gevrey families whose derivatives all vanish at the two ends \cite{elgart2012,ge2016,nenciu1993}, and exponentially small leakage for uniformly gapped Hamiltonians analytic in time \cite{joye1991}. The comparison caps what a product inherits at $O(1/T)$. Where the values of $\sched$ commute, the comparison is the trapezoidal rule, whose Euler--Maclaurin error decays faster than every power of $T$ when all derivatives of $\sched$ vanish at the two ends. Without commutation, do corrections built from the higher Taylor terms of $\sched(t/T)$ give a unit-step comparison error $O(T^{-p})$ for every $p$, perhaps by making the product the evolution of a corrected Hamiltonian $\sched+T^{-1}\Phi_1+\cdots+T^{-k}\Phi_k$? Such a comparison would also give, from the continuous theorems, the higher-order discrete adiabatic theorem of Dranov, Kellendonk, and Seiler \cite{dranov1998}, in whose proof An, Costa, and Berry \cite{an2025} report a missing step.
\item \emph{Finite-precision entries.} The clock error $t-\lfloor t+1/2\rfloor$ changes sign within each step, so, after rescaling, \eqref{eq:products} bounds the time-discretization error uniformly in $T$ by a quantity that tends to zero with the step length, as do the Trotter-error terms of the ground-state bounds of Kovalsky et al.\ \cite{kovalsky2023} and Lu et al.\ \cite{lu2026}. Truncation error does not change sign, and \eqref{eq:fine-bits} lets it grow linearly in $T$. Rounding to the nearest value cannot help by itself, as the evolutions generated by $\operatorname{diag}(x,-x)$ and $\operatorname{diag}(y,-y)$ differ at time $T$ by $2|\sin(T(x-y)/2)|$ in norm. Does rounding in alternating directions, a phase correction, or a weaker error measure such as leakage bound the entry error uniformly in $T$?
\item \emph{Open systems.} \Cref{thm:comparison} assumes Hermitian generators, and every bound drawn from it uses unitary invariance. Does the unit-step comparison extend to contracting evolutions, so that the adiabatic theorems of Avron et al.\ \cite{avron2012} transfer to products of quantum channels?
\item \emph{Sharp nonresonance.} A resonance matters only where $\sched'$ couples the two resonant eigenvectors, and the spin of \Cref{ex:precessing-spin}(ii) shows that the obstruction is then real. Is the sharp hypothesis nonresonance on coupled pairs only, measured by the time $|\{s:\operatorname{dist}(x(s),2\pi\mathbb Z\setminus\{0\})<\varepsilon\}|$ that each such eigenvalue difference $x(s)$ spends near resonance, and does the same hold for the free arc of \Cref{cor:walks}?
\item \emph{Arbitrary sequences.} The quantum approximate optimization algorithm (QAOA) of Farhi, Goldstone, and Gutmann \cite{farhi2014} applies layers $e^{-i\beta_jB}e^{-i\gamma_jC}$, which they relate to adiabatic evolution through a Trotterized evolution with small angles and many layers. When the angles follow a smooth schedule, \Cref{cor:walks} compares the circuit itself with an evolution, but optimized angles need not follow a schedule. Costa et al.\ \cite[Theorem~3]{costa2022} bound the distance of any finite sequence of unitaries from ideal adiabatic evolution by its first and second differences, assuming a gap across every three consecutive unitaries and first differences small compared with it. Does a bound of the same form hold under a gap condition on consecutive pairs alone, so that it applies to QAOA with any small angle steps?
\end{enumerate}

\subsection{Outlook}
Read as an evolution on a rounded clock, a product of unitaries is not a separate object with a theory of its own. It is the continuous evolution at finite resolution in time, and, with a half-step phase at each end, the comparison bounds the distance between the two. Yi and Crosson \cite{yicrosson2021v1} apply continuous adiabatic theory to logarithms of Trotter steps. When the Trotter steps are samples of a $C^2$ unitary family whose spectrum avoids an arc of fixed length that moves continuously, \Cref{cor:walks} is the bound on the change from time discretization that their argument lacked. The same route is open wherever a product of unitaries is read as the unit-step product of a time-dependent Hamiltonian, and heuristics built from such products, QAOA among them, are where I expect it to reach furthest. I anticipate that the comparison will be useful in settings beyond the scope of this paper, including other forms of Hamiltonian simulation and near-term algorithms.
\appendix
\crefalias{section}{appendix}\crefalias{subsection}{subappendix}%
\section{Lemmas behind the unit-step comparisons}\label{app:spectral}
For $0\le\width<2\pi$, let $\Omega$ be a Hermitian matrix with eigenvalue differences of modulus at most $\width$, and let $[\Omega,\cdot]^n$ be the $n$-fold commutator with $\Omega$. For the least eigenvalue $\mu$ of $\Omega$, $\|\Omega-\mu-\width/2\|\le\width/2$, so $\|[\Omega,A]\|=\|[\Omega-\mu-\width/2,A]\|\le\width\|A\|$ for every matrix $A$. Hence a power series $g(x)=\sum_{n\ge0}g_nx^n$ with every $g_n\ge0$ and radius of convergence at least $2\pi$ defines $g([\Omega,\cdot])A=\sum_{n\ge0}g_n[\Omega,\cdot]^nA$, with
\begin{equation}\label{eq:series-bound}
 \|g([\Omega,\cdot])A\|\le g(\width)\|A\|,
\end{equation}
and, as $([\Omega,\cdot]^nA)_{pq}=(\lambda_p-\lambda_q)^nA_{pq}$ for $A_{pq}=\psi_p^\dagger A\psi_q$ as in \Cref{sec:coarse}, $g([\Omega,\cdot])A$ has $(p,q)$ entry $g(\lambda_p-\lambda_q)A_{pq}$. By the partial fractions of the cotangent \cite[4.22.3]{dlmf}, $\cormult(x)=\sum_{m\ge1}2/(4\pi^2m^2-x^2)$ is such a series, and the two meanings of $\cormult([\Omega,\cdot])A$ agree. For $\Omega(s)$ and $A(s)$ continuously differentiable in $s\in[0,1]$, with every $\Omega(s)$ Hermitian with eigenvalue differences of modulus at most $\width$, termwise differentiation \cite[(8.6.4)]{dieudonne1960} gives
\begin{equation}\label{eq:series-derivative}
\begin{aligned}
 (g([\Omega,\cdot])A)'&=g([\Omega,\cdot])A'+\sum_{n\ge1}g_n\sum_{k=0}^{n-1}[\Omega,\cdot]^k[\Omega',[\Omega,\cdot]^{n-1-k}A],\\
 \|(g([\Omega,\cdot])A)'\|&\le g(\width)\|A'\|+2g'(\width)\|\Omega'\|\,\|A\|,
\end{aligned}
\end{equation}
as each of the $n$ terms with coefficient $g_n$ has norm at most $2g_n\width^{n-1}\|\Omega'\|\,\|A\|$.

The first lemma constructs the correction of \Cref{sec:comparison}.
\begin{lemma}[\lean{UnitaryProducts/PaperProof/LocalCell.lean}{117}]\label{lem:cell}
Let $0\le\width<2\pi$, let $\Omega$ be a Hermitian matrix with eigenvalue differences of modulus at most $\width$, and let $D$ be a matrix. There is a continuously differentiable matrix-valued $Y(\tau)$ for $|\tau|\le\frac12$ such that the following hold.
\begin{enumerate}
\item $Y$ solves \eqref{eq:corrector}, with end values $\cormult([\Omega,\cdot])D$.
\item $\|Y(0)\|\le c_0\|D\|$ for the $c_0$ of \Cref{app:constants}, and $\|Y(\tau)\|\le\|Y(0)\|+\tau^2\|D\|/2$ for $|\tau|\le\frac12$.
\end{enumerate}
\end{lemma}

\begin{proof}
With $A_{pq}=\psi_p^\dagger A\psi_q$ and $x=\lambda_p-\lambda_q$ as in \Cref{sec:coarse}, every matrix $A$ satisfies $[\Omega,A]_{pq}=xA_{pq}$ and, as in \Cref{def:unit-objects}, $(\phi([\Omega,\cdot])A)_{pq}=\phi(x)A_{pq}$ for a function $\phi$. We guess that $Y(\tau)=\cormult_\tau([\Omega,\cdot])D$ solves \eqref{eq:corrector}, where
\[
 \cormult_\tau(x)=\begin{cases}
  \dfrac{1-ix\tau}{x^2}-\dfrac{e^{-ix\tau}}{2x\sin(x/2)}, & x\ne0,\\[2ex]
  \dfrac{\tau^2}2-\dfrac1{24}, & x=0.
 \end{cases}
\]
Differentiating gives $\partial_\tau\cormult_\tau(x)+ix\cormult_\tau(x)=\tau$, so $Y'+i[\Omega,Y]=\tau D$, and $\cormult_{\pm1/2}(x)=\cormult(x)$, so $Y(\pm\frac12)=\cormult([\Omega,\cdot])D$. The two identities verify the guess and prove (i).

The partial fractions of the cosecant \cite[4.22.5]{dlmf} give $-\cormult_0(x)=\sum_{m\ge1}(-1)^{m+1}2/(4\pi^2m^2-x^2)$, whose coefficient of each $x^{2k}$ is an alternating sum with decreasing terms and so positive. Hence, by \eqref{eq:series-bound},
\begin{align*}
 \|Y(0)\|&=\|\cormult_0([\Omega,\cdot])D\|\\
 &\le-\cormult_0(\width)\|D\|\\
 &=c_0\|D\|.
\end{align*}
As $Y'+i[\Omega,Y]=\tau D$, the conjugate $e^{i\tau\Omega}Ye^{-i\tau\Omega}$ has derivative $\tau e^{i\tau\Omega}De^{-i\tau\Omega}$, and integrating from $0$ to $\tau$ gives $Y(\tau)=e^{-i\tau\Omega}\bigl(Y(0)+\int_0^\tau\sigma e^{i\sigma\Omega}De^{-i\sigma\Omega}\,d\sigma\bigr)e^{i\tau\Omega}$. Hence
\begin{align*}
 \|Y(\tau)\|&=\Bigl\|Y(0)+\int_0^\tau\sigma e^{i\sigma\Omega}De^{-i\sigma\Omega}\,d\sigma\Bigr\|\\
 &\le\|Y(0)\|+\frac{\tau^2}2\|D\|.
\end{align*}
\end{proof}

The second lemma constructs the logarithm of \Cref{sec:walks} and bounds its derivatives.
\begin{lemma}[\lean{UnitaryProducts/PaperProof/LemWalkLog.lean}{552}]\label{lem:walk-log}
Let $0\le\width<2\pi$, let $a\colon[0,1]\to\mathbb R$ be continuous, and let $\devprod(s)$ for $s\in[0,1]$ be a $C^2$ unitary family with $\operatorname{spec}\devprod(s)\subset\{e^{-i\theta}:\theta\in[a(s),a(s)+\width]\}$. Then $\devprod=e^{-i\sched}$ for exactly one Hermitian $\sched(s)$ with spectrum in $[a(s),a(s)+\width]$, $\sched$ is twice continuously differentiable in $s\in[0,1]$, and, with the $\kappa$ of \Cref{app:constants}, for every $s\in[0,1]$,
\begin{enumerate}
\item $\|\sched'(s)\|\le\kappa\|\devprod'(s)\|$, and
\item $\|\sched''(s)\|\le\kappa\|\devprod''(s)\|+\|\devprod'(s)\|^2/\cos^2(\width/4)$.
\end{enumerate}
\end{lemma}

\begin{proof}
For each $s\in[0,1]$, the spectral theorem and the hypothesis on $\operatorname{spec}\devprod(s)$ give $\devprod(s)=\sum_je^{-i\theta_j}P_j$ with spectral projectors $P_j$ and $\theta_j\in[a(s),a(s)+\width]$, so $\sched(s)=\sum_j\theta_jP_j$ is Hermitian with spectrum in $[a(s),a(s)+\width]$ and $e^{-i\sched(s)}=\devprod(s)$. As $\theta\mapsto e^{-i\theta}$ is injective for $\theta\in[a(s),a(s)+\width]$, $\sched(s)$ is the only Hermitian matrix with spectrum in $[a(s),a(s)+\width]$ and $e^{-i\sched(s)}=\devprod(s)$. Fix $s_0\in[0,1]$. Multiplying $\devprod$ by the constant phase $e^{i(a(s_0)+\width/2)}$ shifts $\sched$ by a constant and changes none of the norms in (i) and (ii), so we may take $a(s_0)=-\width/2$. Then every eigenvalue $\varphi$ of $\sched(s_0)$ satisfies $|\varphi|\le\width/2$, and, as $a$ is continuous and $\width<2\pi$, every eigenvalue $\varphi$ of $\sched(s)$ for $s$ near $s_0$ satisfies $|\varphi|<\pi$, so $\sched(s)=i\operatorname{Log}\devprod(s)$ for $s$ near $s_0$, with $\operatorname{Log}$ the principal logarithm. For complex $z\notin(-\infty,0]$, $\operatorname{Log}z=\int_0^1(z-1)(1+u(z-1))^{-1}\,du$. Applying this integral formula to $\devprod(s)$, with $S_u=(1+u(\devprod-1))^{-1}$, and differentiating twice,
\begin{equation}\label{eq:log-integral}
\begin{aligned}
 \sched&=i\int_0^1(\devprod-1)S_u\,du,\\
 \sched'&=i\int_0^1S_u\devprod'S_u\,du,\\
 \sched''&=i\int_0^1\bigl(S_u\devprod''S_u-2uS_u\devprod'S_u\devprod'S_u\bigr)\,du,
\end{aligned}
\end{equation}
where the second and third lines follow from the first and second by $u(\devprod-1)S_u=1-S_u$ and $\partial_sS_u=-uS_u\devprod'S_u$, so $\sched$ is $C^2$ for $s$ near $s_0$.

For $0\le u\le1$, let $v=2u-1$ and $r(v)=\cos^2(\width/4)+v^2\sin^2(\width/4)$. For every eigenvalue $\varphi$ of $\sched(s_0)$, the number $1+u(e^{-i\varphi}-1)=e^{-i\varphi/2}(\cos(\varphi/2)-iv\sin(\varphi/2))$ has squared modulus $1-(1-v^2)\sin^2(\varphi/2)\ge r(v)$, as $|\varphi|\le\width/2$, so $\|S_u\|^2\le1/r(v)$ at $s_0$. By \eqref{eq:log-integral} at $s_0$, with $u=(1+v)/2$,
\begin{align*}
 \|\sched'\|&\le\frac{\|\devprod'\|}2\int_{-1}^1\frac{dv}{r(v)}\\
 &=\kappa\|\devprod'\|,\\
 \|\sched''\|&\le\frac{\|\devprod''\|}2\int_{-1}^1\frac{dv}{r(v)}+\frac{\|\devprod'\|^2}2\int_{-1}^1\frac{1+v}{r(v)^{3/2}}\,dv\\
 &=\kappa\|\devprod''\|+\frac{\|\devprod'\|^2}2\int_{-1}^1\frac{dv}{r(v)^{3/2}}\\
 &=\kappa\|\devprod''\|+\frac{\|\devprod'\|^2}{\cos^2(\width/4)}.
\end{align*}
The fourth line drops the odd part $v\,r(v)^{-3/2}$ of $(1+v)r(v)^{-3/2}$.
\end{proof}

The third lemma bounds the derivative of $\cormult([\Omega,\cdot])\Omega'$ along a curve $\Omega$. Its bound (ii) replaces one factor $\|\Omega'\|$ by $\|(e^{-i\Omega})'\|$, which is $\|\devprod'\|$ for a unitary family $\devprod=e^{-i\Omega}$, at the cost of the term $(\cormult(\width)-\cormult(0))\|\Omega'\|^2$. Using (i) instead in the proof of \Cref{cor:walks}, with $\|\sched'\|\le\kappa\|\devprod'\|$ of \Cref{lem:walk-log}, would replace the terms $2\cormult'(\width)\kappa+(\cormult(\width)-\cormult(0))\kappa^2$ of the coefficient $c_2$ of \Cref{app:constants} by $2\cormult'(\width)\kappa^2\sim4\pi(2\pi-\width)^{-4}$.

\begin{lemma}[\lean{UnitaryProducts/PaperProof/LocalCell.lean}{208}]\label{lem:end-value}
Let $0\le\width<2\pi$, and let $\Omega(s)$ be twice continuously differentiable for $s\in[0,1]$, with every $\Omega(s)$ Hermitian with eigenvalue differences of modulus at most $\width$. Then $\cormult([\Omega,\cdot])\Omega'$ is continuously differentiable in $s\in[0,1]$, and for every $s\in[0,1]$ the following hold.
\begin{enumerate}
\item $\|(\cormult([\Omega,\cdot])\Omega')'\|\le\cormult(\width)\|\Omega''\|+2\cormult'(\width)\|\Omega'\|^2$.
\item $\|(\cormult([\Omega,\cdot])\Omega')'\|\le\cormult(\width)\|\Omega''\|+2\cormult'(\width)\|(e^{-i\Omega})'\|\,\|\Omega'\|+\bigl(\cormult(\width)-\cormult(0)\bigr)\|\Omega'\|^2$.
\end{enumerate}
\end{lemma}

\begin{proof}
With $g=\cormult$ and $A=\Omega'$, \eqref{eq:series-derivative} gives the continuous differentiability of $\cormult([\Omega,\cdot])\Omega'$ and (i). For (ii), with $\cormult(x)=\sum_{n\ge0}\nu_nx^n$, let
\[
 \mathcal F(A)=\sum_{n\ge1}\nu_n\sum_{k=0}^{n-1}[\Omega,\cdot]^k[A,[\Omega,\cdot]^{n-1-k}\Omega']
\]
for matrices $A$, so that $(\cormult([\Omega,\cdot])\Omega')'=\cormult([\Omega,\cdot])\Omega''+\mathcal F(\Omega')$ and, counting terms as in \eqref{eq:series-derivative}, $\|\mathcal F(A)\|\le2\cormult'(\width)\|A\|\,\|\Omega'\|$. As $\mathcal F$ is linear, and on a commutator $A=[\Omega,\Lambda]$ the sum over $k$ telescopes, with $\cormult(\width)-\cormult(0)$ in place of $\cormult'(\width)$, we write $\Omega'=X+[\Omega,\Lambda]$ for matrices $X$ and $\Lambda$, chosen below with $\|X\|=\|(e^{-i\Omega})'\|$ and $\|\Lambda\|\le\|\Omega'\|/4$. Then
\begin{align}
 \mathcal F([\Omega,\Lambda])&=\sum_{n\ge1}\nu_n\sum_{k=0}^{n-1}\bigl([\Omega,\cdot]^{k+1}[\Lambda,[\Omega,\cdot]^{n-1-k}\Omega']-[\Omega,\cdot]^k[\Lambda,[\Omega,\cdot]^{n-k}\Omega']\bigr)\notag\\
 &=\sum_{n\ge1}\nu_n\bigl([\Omega,\cdot]^n[\Lambda,\Omega']-[\Lambda,[\Omega,\cdot]^n\Omega']\bigr),\notag\\
 \|(\cormult([\Omega,\cdot])\Omega')'\|&\le\cormult(\width)\|\Omega''\|+\|\mathcal F(X)\|+\|\mathcal F([\Omega,\Lambda])\|\notag\\
 &\le\cormult(\width)\|\Omega''\|+2\cormult'(\width)\|X\|\,\|\Omega'\|+4\bigl(\cormult(\width)-\cormult(0)\bigr)\|\Lambda\|\,\|\Omega'\|.\label{eq:F-split}
\end{align}
The first line applies Jacobi's identity $[[\Omega,\Lambda],B]=[\Omega,[\Lambda,B]]-[\Lambda,[\Omega,B]]$ to $B=[\Omega,\cdot]^{n-1-k}\Omega'$ in the $k$-th term of $\mathcal F([\Omega,\Lambda])$, the second telescopes the sum over $k$, the third splits $\mathcal F(\Omega')=\mathcal F(X)+\mathcal F([\Omega,\Lambda])$ and bounds $\cormult([\Omega,\cdot])\Omega''$ by \eqref{eq:series-bound}, and the fourth applies the bound on $\|\mathcal F(A)\|$ with $A=X$ and bounds each of the two terms in the second line by $2(\cormult(\width)-\cormult(0))\|\Lambda\|\,\|\Omega'\|$, by \eqref{eq:series-bound} for $\cormult(x)-\cormult(0)=\sum_{n\ge1}\nu_nx^n$.

To choose $X$ and $\Lambda$, let $\Omega'_\sigma=e^{i\sigma\Omega}\Omega'e^{-i\sigma\Omega}$, so $[\Omega,\Omega'_\sigma]=-i\partial_\sigma\Omega'_\sigma$, and let
\begin{align*}
 X&=\int_{-1/2}^{1/2}\Omega'_\tau\,d\tau,\\
 \Lambda&=-i\int_{-1/2}^{1/2}\int_0^\tau\Omega'_\sigma\,d\sigma\,d\tau.
\end{align*}
Then $\Omega'=X+[\Omega,\Lambda]$, as
\begin{align*}
 [\Omega,\Lambda]&=-i\int_{-1/2}^{1/2}\int_0^\tau[\Omega,\Omega'_\sigma]\,d\sigma\,d\tau\\
 &=-\int_{-1/2}^{1/2}\int_0^\tau\partial_\sigma\Omega'_\sigma\,d\sigma\,d\tau\\
 &=-\int_{-1/2}^{1/2}(\Omega'_\tau-\Omega')\,d\tau\\
 &=\Omega'-X.
\end{align*}
Moreover, $\|X\|=\|(e^{-i\Omega})'\|$, as $(e^{-i\Omega})'=-i\int_0^1e^{-i\sigma\Omega}\Omega'e^{-i(1-\sigma)\Omega}\,d\sigma=-ie^{-i\Omega/2}Xe^{-i\Omega/2}$ with $\sigma=1/2-\tau$, and $\|\Lambda\|\le\int_{-1/2}^{1/2}|\tau|\,d\tau\,\|\Omega'\|=\|\Omega'\|/4$, so \eqref{eq:F-split} is (ii).
\end{proof}

\section{Constants}\label{app:constants}
For $0<\width<2\pi$, let $c_0=1/(2\width\sin(\width/2))-1/\width^2$ and $\kappa=(\width/2)/\sin(\width/2)$, and at width $0$ let $c_0=1/24$ and $\kappa=1$, their limits as $\width\to0$. Let $c_1=c_0/4+1/64$ and $c_2=(1/24+\cormult(\width))/\cos^2(\width/4)+c_1\kappa^2+2\cormult'(\width)\kappa+(\cormult(\width)-\cormult(0))\kappa^2$. For $\sched$, $L$, and $\width$ as in \Cref{thm:clocks}, let
\[
 \mathcal E_\width(\sched)=c_0\bigl(\|\sched'(0)\|+\|\sched'(1)\|\bigr)+\Bigl(\frac1{24}+\cormult(\width)\Bigr)L+\bigl(c_1+2\cormult'(\width)\bigr)\|\sched'\|_\infty^2.
\]
Its three terms come from the two ends, from the curvature of $\sched$, and from the parts of the residual and of the jumps that are quadratic in $\sched'$. For $\devprod$ and $\width$ as in \Cref{cor:walks}, let
\[
 \widetilde{\mathcal E}_\width(\devprod)=c_0\kappa\bigl(\|\devprod'(0)\|+\|\devprod'(1)\|\bigr)+\Bigl(\frac1{24}+\cormult(\width)\Bigr)\kappa\|\devprod''\|_\infty+c_2\|\devprod'\|_\infty^2.
\]

The width-$\pi$ values give \Cref{introthm:unit}.
\begin{fact}[\lean{UnitaryProducts/PaperProof/ExamplesPi.lean}{105}]\label{fact:pi}
At width $\pi$, $c_0=1/(2\pi)-1/\pi^2$, $\cormult(\pi)=1/\pi^2$, $\cormult'(\pi)=1/(4\pi)-2/\pi^3$, and
\[
 \mathcal E_\pi(\sched)\le0.116\|\sched'\|_\infty+0.143L+0.061\|\sched'\|_\infty^2.
\]
\end{fact}

\begin{proof}
The three values follow from the definitions and \eqref{eq:N-def}. The bound follows from $\|\sched'(0)\|+\|\sched'(1)\|\le2\|\sched'\|_\infty$, $2c_0<0.116$, $1/24+\cormult(\pi)<0.143$, and $c_1+2\cormult'(\pi)<0.061$.
\end{proof}

The coefficients for products of given unitaries grow as the free arc closes.
\begin{fact}[\lean{UnitaryProducts/PaperProof/CorProducts.lean}{304}]\label{fact:arc}
As $\width\to2\pi$, with $\freearc=2\pi-\width$, $c_0\kappa\sim(1/24+\cormult(\width))\kappa\sim\freearc^{-2}$ and $c_2\sim(8/\pi+5\pi/2+2)\freearc^{-3}$, so every coefficient of $\widetilde{\mathcal E}_\width(\devprod)$ is $O(\freearc^{-3})$.
\end{fact}

\begin{proof}
As $\width\to2\pi$, $\sin(\width/2)=\sin(\freearc/2)\sim\freearc/2$ and $\cos(\width/4)=\sin(\freearc/4)\sim\freearc/4$, while $\cormult(\width)\sim1/(2\pi\freearc)$ and $\cormult'(\width)\sim1/(2\pi\freearc^2)$ by \eqref{eq:N-def}. Hence $\kappa\sim2\pi/\freearc$, $c_0\sim\cormult(\width)$, and $c_1\sim\cormult(\width)/4$, so $c_0\kappa\sim(1/24+\cormult(\width))\kappa\sim\freearc^{-2}$, and
\[
 c_2\sim\frac{\cormult(\width)}{\cos^2(\width/4)}+\frac{\cormult(\width)}4\kappa^2+2\cormult'(\width)\kappa+\cormult(\width)\kappa^2\sim\frac{8/\pi+5\pi/2+2}{\freearc^3}.
\]
\end{proof}
 \backmatter
\section*{Acknowledgments}
I thank Prabal Adhikari, Evelyn Sander, Jacob Bringewatt, Sam Mendelson, and Geir Agnarsson for
helpful conversations.

I used AI tools (Anthropic's Claude Code, OpenAI's Codex, and Harmonic's Aristotle) under my direction for three tasks: to formalize proofs in Lean, to survey and compare the literature, and to draft and revise the text. During revision, I worked out some proofs of \Cref{sec:coarse} and \Cref{app:spectral} together with these tools, and those proofs differ from my original arguments.

I checked every proof by hand. As an error introduced into a familiar argument is easy to miss on reading, every main claim is also checked in the Lean~4 proof assistant with the Mathlib library. Each numbered statement of \Cref{sec:prelim,sec:results,sec:coarse} and \Cref{app:spectral}, and each of \Cref{introthm:fine-count,introthm:unit}, links from its header to its Lean declaration in the public repository \cite{jarret2026lean}, and the Lean proofs include the Davis--Kahan theorem. The continuous theorems that the transfers take as input, such as the theorem of Jansen, Ruskai, and Seiler, enter the Lean statements as hypotheses and are not checked there.
\section*{Statements and Declarations}
\begin{description}
\item[Funding.] I received no funding for this work.
\item[Competing interests.] I have no competing interests to declare.
\item[Ethics approval and consent to participate.] Not applicable.
\item[Consent for publication.] Not applicable.
\item[Data availability.] I do not analyze or generate any datasets, because my work proceeds within a theoretical and mathematical approach.
\item[Code availability.] The Lean~4 code is available at \leanrepolink{} \cite{jarret2026lean}, and the statement headers link to its release v1.0.0.
\item[Author contributions.] I am the sole author of this paper.
\item[Use of AI tools.] The Acknowledgments describe how I used AI tools. No AI tool is an author of this paper, and I am accountable for its entire content.
\end{description}

\end{document}